\documentclass[11pt, reqno, a4paper]{amsart}
\usepackage[table,xcdraw]{xcolor}
\usepackage{braket}
\usepackage{graphicx}
\usepackage{tikz}
\usepackage{amssymb,amsthm}
\usepackage[margin=1in]{geometry}
\usepackage[colorlinks = true, linkcolor = blue, urlcolor  = blue, citecolor = red]{hyperref}
\usepackage{subfig}
\renewcommand{\epsilon}{\varepsilon}
\renewcommand{\phi}{\varphi}
\newcommand{\C}{\mathbb{C}}
\newcommand{\R}{\mathbb{R}}

\newcommand{\M}{\mathcal{M}}
\DeclareMathOperator{\Tr}{Tr}

\newcommand{\ketbra}[2]{|#1\rangle\langle#2|} 

\newtheorem{theorem}{Theorem}[section]
\newtheorem{definition}[theorem]{Definition}
\newtheorem{proposition}[theorem]{Proposition}

\newtheorem{lemma}[theorem]{Lemma}
\newtheorem{example}[theorem]{Example}

\newtheorem{remark}[theorem]{Remark}

\begin{document}

	\title[Measurement incompatibility vs. Bell non-locality]{Measurement incompatibility vs. Bell non-locality:\\an approach via tensor norms}
	
	\author{Faedi Loulidi}
	
	\author{Ion Nechita}
	\email{$\{$loulidi,nechita$\}$@irsamc.ups-tlse.fr}
	\address{Laboratoire de Physique Th\'eorique, Universit\'e de Toulouse, CNRS, UPS, France}	
	
	\begin{abstract}
		Measurement incompatibility and quantum non-locality are two key features of quantum theory. Violations of Bell inequalities require quantum entanglement and incompatibility of the measurements used by the two parties involved in the protocol. We analyze the converse question: for which Bell inequalities is the incompatibility of measurements enough to ensure a quantum violation? We relate the two questions by comparing two tensor norms on the space of dichotomic quantum measurements: one characterizing measurement compatibility and the second one characterizing violations of a given Bell inequality. We provide sufficient conditions for the equivalence of the two notions in terms of the matrix describing the correlation Bell inequality. We show that the CHSH inequality and its variants are the only ones satisfying it. 
	\end{abstract}
	
	\date{\today}
	
	\maketitle
	
	\tableofcontents
	
	\section{Introduction}
	Since its discovery, quantum mechanics was formalized as a theory with many foundational aspects which differ significantly from classical mechanics. Some of these deep questions, and the relation among them are still subject to investigation nowadays. Understanding these notions and their interplay is crucial for the development of the second quantum revolution.
	
	Two of the most important conceptual revolutions put forward by quantum mechanics are the notions of \emph{non-locality of correlations} and the \emph{incompatibility of quantum measurements}. The latter notion, that of measurement incompatibility is one of the most unintuitive aspects of the quantum world, when examined from a classical perspective: there exist (quantum) measurements which cannot be performed simultaneously on a given quantum system. 
	
	It is well-known that quantum non-locality is one of the fundamental aspects of quantum theory that gives rise to a lot of questions about quantum reality. John Bell \cite{Bell} gave a complete answer to the debate about the non-locality and elucidated the intrinsic probabilistic aspect of quantum theory. The answer he provides is that any local theory must obey some inequality, while if one applies the predictions of quantum mechanics, the aforementioned inequality can be violated. This means that the quantum world is completely non-local, which, in turn, means that there are phenomena that we could not understand with our classical macroscopic point of view. Such conclusion provides a complete answer about the intrinsic reality of the quantum world. Such violations of correlation inequalities were completely confirmed experimentally in Alain Aspect’s experiment \cite{Aspect}, and in a loophole-free manner in \cite{hensen2015loophole}. 
	
	In the modern language of quantum information theory, such correlation inequalities can be understood as \emph{non-local games} \cite{palazuelos2016survey}. In such a game, two players, called traditionally Alice and Bob, play cooperatively against a Referee. Alice and Bob are space-like separated, hence once the game starts they can no longer communicate. However, they both know the rules of the game, and they can meet before the game starts and make a strategy. Technically, the games we consider are defined by a matrix $M$, which encodes the pay-off the players receive; in particular, we shall consider in this work exclusively XOR games with $N$ questions and two answers. 
	
	Two scenarios are of particular importance for us: Alice and Bob are either allowed to use \emph{classical strategies} (where they share classical randomness) or \emph{quantum strategies} (where they share a bipartite entangled quantum state). It turns out that the optimal probabilities to win the game with classical or, respectively, quantum strategies, can be formulated as two different \emph{tensor norms} of the matrix that encodes the rules of the game (seen as a $2-$tensor). This is one of the instances where tensor norms (and Banach space theory in general) has found applications in the theory of non-local games. 
	
	The main goal of the current work is to relate, in a \emph{quantitative manner} the notion of \emph{measurement incompatibility} to that of \emph{Bell inequality violations} in a very general setting. Our original motivation was the seminal work \cite{wolf2009measurements}, where the authors connected, in a qualitative manner, the incompatibility of Alice's measurements in the CHSH game, with possible violations of the CHSH inequality. Our results can be seen to build on this example, generalizing it in two different directions: 
	\begin{itemize}
		\item we go beyond the CHSH game, allowing all (reasonable) correlation XOR games
		\item we relate, in a quantitative manner, the largest possible violation of a Bell inequality to the incompatibility robustness of Alice's measurements. 
	\end{itemize}
	In order to achieve these goals, our framework is different than the usual setting of non-local games, in the respect that
	\begin{center}
		\boxed{
			\emph{Alice's dichotomic measurements are fixed.}}
	\end{center}
	Optimizing over Bob's choice of $N$ dichotomic measurements and over the players' shared entangled state, we can express the quantum bias of the given non-local game $M$ as a tensor norm of Alice's $N$-tuple of measurements, which we denote by $\|A\|_M$. In particular, with Alice's choice of measurements fixed to be $A$, the players will violate the Bell inequality corresponding to $M$ if and only if $\|A\|_M > \beta(M)$, where $\beta(M)$ is the classical bias of the game.
	
	On the quantum measurement (in-)compatibility side, we revisit the construction from \cite{bluhm2022incompatibility}, where the compatibility of a $N$-tuple of dichotomic quantum measurements has been described with the help of a tensor norm, dubbed $\|A\|_c$. We give a direct proveof the result showing that $N$ dichotomic measurements $A = (A_1, \ldots, A_N)$ are compatible iff $\|A\|_{c}\leq 1$. The value of the compatibility norm $\|\cdot\|_c$ is related to the notion of \emph{compatibility robustness}: the value of the norm is precisely the quantity of (white) noise one needs to add to the tuple of dichotomic measurements in order to render them compatible. 
	
	Having formulated the key physical principles of this work (quantum incompatibility and Bell non-locality), we get now to our main point: the relation between them. This question has received already a lot of attention in the literature. The starting point is the equivalence first observed in \cite{wolf2009measurements}: for the CHSH game \cite{clauser1969proposed} with two questions, Alice's pair of measurements are incompatible if and only if there exist an entangled state and a choice for Bob's pair of measurements such that they can obtain a violation of the CHSH Bell inequality. It is equally well-known that the two notions are not equivalent in more general situations,  see \cite{quintino2014joint}. 
	
	In this work, we provide a definitive answer to this question, using the framework of tensor norms. More precisely, we express the following quantities as tensor norms: 
	\begin{itemize}
		\item \emph{incompatibility}: how much (white) noise one needs to add to a tuple of dichotomic POVMs to render them compatible
		\item \emph{quantum bias of a correlation game}: what is the maximal value of the game (normalized to have classical bias 1), when Alice's tuple of dichotomic measurements are fixed.
	\end{itemize}
	
	We then discuss how these norms compare, and when they are equal. We provide sufficient conditions for equality, and then show that only the CHSH game (and its permutations) satisfy them, emphasizing the special role of the CHSH inequality. 
	
	\begin{figure}
		\centering
		\includegraphics[width=\textwidth]{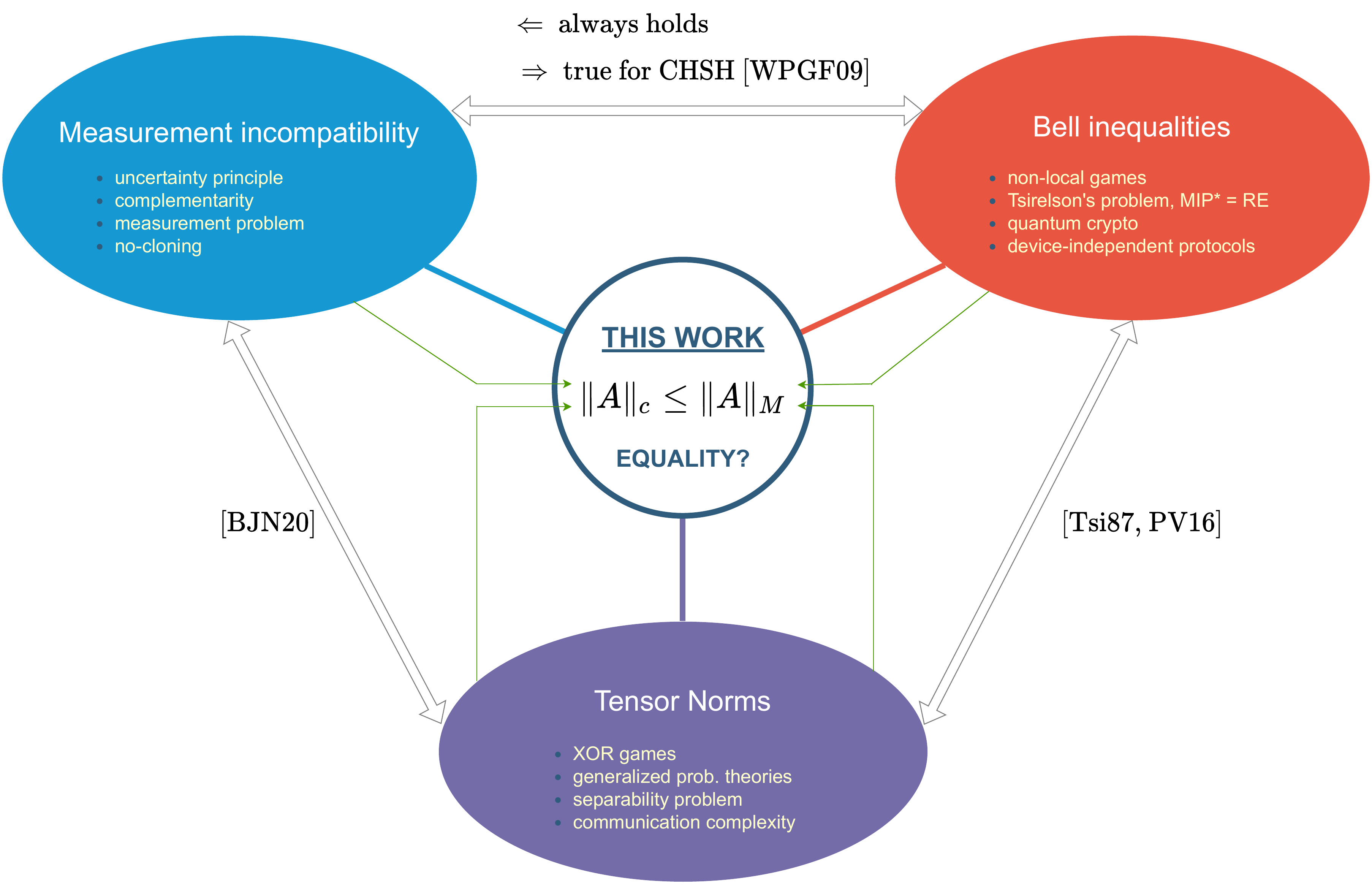}
		\caption{In this work, we relate measurement incompatibility with Bell inequalities using the formalism of tensor norms. Pairwise connections having already been established in the literature, we bring the three concepts together for the first time.}
		\label{fig:non-locality}
	\end{figure}
	
	Our paper is organised as follows. In Section \ref{mainresult} we (informally) state the main results of our paper and their interpretation. In Section \ref{sec:compatibility} we recall the notion of compatibility for quantum measurements. We present in Section \ref{sec:Tensor Norms} the basic definitions of tensor norms from Banach space theory, focusing on the examples needed in this work. In Section \ref{sec:non-locality} we introduce the framework of Bell non-locality as non-local games and relate the values of these games to tensor norms. In Section \ref{sec:non-locality-norm} we introduce the main definition of the non-locality norm $\|A\|_M$ that will characterise the violation of the Bell inequality. In Section \ref{sec: compatibility-norm} we introduce the compatibility norm $\|A\|_{c}$ that will characterise the compatibility of Alice's measurement. We present in the Section \ref{sec: non-locality and incompatibility} our main theorems, discussing also under which conditions the violation of a Bell inequality implies measurement incompatibility. In our framework, we provide a conceptual explanation of the main result in \cite{wolf2009measurements}, and we also analyze new Bell inequalities, such as different deformations of the CHSH inequality and the pure correlation part of the $I_{3322}$ tight Bell inequality; for the latter, the two notions are not equivalent, as noticed in \cite{quintino2014joint}.    
	
	\section{Main results}\label{mainresult}
	In this section we introduce the main definitions and the main results of our work. Our goal is to unify two fundamental notions of quantum theory, \emph{measurement incompatibility} and \emph{Bell inequality violations}. To do so, we shall work in the framework of \emph{non-local games}, where the rules of a correlation game are encoded in a real $N \times N$ matrix $M$, and \emph{Alice's dichotomic measurements are fixed}. Note that in this work we shall be considering general (not necessarily projective) measurements, mathematically encoded by POVMs. 
	
	The maximum value of the game $M$, when Alice's measurements are fixed, is given by the following quantity. 
	
	\begin{definition}[The $M$-Bell-locality tensor norm] Let $M$ an invertible \emph{Bell functional}  and Alice's $N$-tuple of dichotomic measurements $A = (A_1, \ldots, A_N)$, we define the following tensor norm: 
		$$\|A\|_{M}:=\sup_{\|\psi\| = 1} \sup_{\|B_y\| \leq 1} \Big \langle \psi \Big | \sum_{x,y = 1}^N\,M_{xy} \,A_x \otimes B_y \Big | \psi \Big \rangle=\lambda_{\max}\bigg[\sum_{y=1}^N \bigg| \sum_{x=1}^N  M_{xy}\,A_x\bigg|\bigg].$$
	\end{definition}
	The quantity $\|A\|_M$ is the maximum value of the game $M$, when optimizing over quantum strategies, with Alice's measurements being fixed. The measurements $A = (A_1, \ldots, A_N)$ are called \emph{$M$-Bell-local} there is no violation of the Bell inequality corresponding to $M$: $\|A\|_M\leq \beta(M)$, with $\beta(M)$ being \emph{the classical bias} of the game (which, importantly, can also be expressed as a tensor norm). If this is not the case, we call Alice's measurements \emph{$M$-Bell-non-local}.
	
	Regarding compatibility, we are concerned with the same question as before: are Alice's dichotomic measurements compatible or not? The following quantity was introduced, in the more abstract setting of generalized probabilistic theories in \cite{bluhm2022incompatibility}, see also \cite{bluhm2022tensor}. 
	
	\begin{definition}[The compatibility tensor norm]
		For a tensor $A \in  \mathbb R^N \otimes \mathcal M_d^{sa}(\mathbb C)$, we define the following quantity: 
		\begin{equation*}
			\|A\|_c := \inf \Bigg\{ \Big\|\sum_{j=1}^K H_j\Big\|_\infty \, : \, A = \sum_{j=1}^K z_j \otimes H_j, \, \text{ s.t. } \, \forall j \in [K], \, \|z_j\|_\infty \leq 1 \text{ and } H_j \geq 0\Bigg\}.
		\end{equation*}
	\end{definition}
	
	The compatibility norm, together with the injective tensor product of $\ell_\infty$ and $S_\infty$ norms, completely characterize compatibility of tuples of dichotomic quantum measurements \cite{bluhm2022incompatibility,bluhm2022tensor}. 
	
	\begin{proposition}
		Let $A = (A_1, \ldots, A_N)$ be a $N$-tuple of self-adjoint $d \times d$ complex matrices. Then:
		\begin{enumerate}
			\item $A$ is a collection of dichotomic quantum observables (i.e.~$\|A_i\|_\infty \leq 1$ $\forall i$) if and only if $\|A\|_\epsilon \leq 1$.
			\item $A$ is a collection of \emph{compatible} dichotomic quantum observables if and only if $\|A\|_c \leq 1$.
		\end{enumerate}
	\end{proposition}
	The compatibility norm allows Alice to know whether her measurements are compatible ($\|A\|_c\leq 1$) or not ($\|A\|_c > 1$); in the latter case, the, the minimal quantity of white noise that needs to be mixed in the measurements in order to render them compatible is $1/\|A\|_c$, providing an operational interpretation of the compatibility norm.
	
	To sum up, in the setting  of tensor norms, 
	\begin{itemize}
		\item Alice's measurements are \emph{$M$-Bell-local} if and only if $\|A\|_M\leq\beta (M)=\|M\|_{\ell_1^N \otimes_{\epsilon} \ell_1^N}.$
		\item Alice's measurements are compatible if and only if $\|A\|_c\leq 1$.
	\end{itemize}
	
	To understand the relation between non-locality and the compatibility, we now have to compare the two norms $\|\cdot\|_c$ and $\|\cdot\|_M$.
	
	\begin{theorem}
		Consider an $N$-input 2-output non-local game $M$, corresponding to a matrix $M \in \mathcal M_N(\mathbb R)$. Then, for any $N$-tuple of self-adjoint matrices $A = (A_1, \ldots, A_N)$, we have
		$$\|A\|_M\leq\|A\|_{c} \, \beta(M)=\|A\|_{c}\,\|M\|_{\ell_1^N \otimes_{\epsilon} \ell_1^N}.$$
		In particular, if Alice's measurements $A$ are $M$-Bell-non-local, then they must be incompatible. 
	\end{theorem}
	
	In the theorem above, we have upper bounded the $M$-Bell-locality norm by the compatibility norm that depends only on Alice's measurement times the classical bias of the game. This inequality is a \emph{quantitative} version of the well-known, \emph{qualitative} fact that if Alice's measurement are compatible, she will never observe any Bell inequality violation (i.e.~her measurements are $M$-Bell-local). 
	
	One of our main contributions is to raise and answer the converse question: we want to upper bound the compatibility norm by the $M$-Bell-locality norm. In physical terms, we are asking whether, given a Bell inequality $M$ and a tuple of measurements, can Alice observe violations of $M$ using her measurements? We have the following theorem, providing a (partial) answer to this question.
	
	\begin{theorem}
		Let $M \in \mathcal M_N(\mathbb R)$ be an invertible matrix. Then, for any $N$-tuple of self-adjoint matrices $A = (A_1, A_2, \ldots, A_N)$, we have
		$$\|A\|_c \leq \|A\|_M\|M^{-1}\|_{\ell_\infty^N \otimes_\epsilon \ell_\infty^N}.$$
	\end{theorem}
	
	In the main theorems succinctly stated above, we have compared the compatibility tensor norm and the  $M$-Bell-locality norm. It was shown in \cite{wolf2009measurements} that for the CHSH game, the incompatibility of one party's quantum measurements and the violation of a Bell inequality are equivalent. In our setting, this equivalence can be understood as an \emph{equality} of the compatibility norm and the $M$-Bell-locality norm for $M_{\text{CHSH}}$: we have $$\|\cdot\|_c=\|\cdot\|_{M_{\text{CHSH}}}.$$
	
	Having restated this classical result in terms of an equality of tensor norms, it is natural to ask whether this equality goes beyond the case of the CHSH inequality. Incompatibility and Bell non-locality are not, in general, equivalent, as it was shown in \cite{bene2018measurement, hirsch2018quantum}.
	
	From the main theorems above, any game $M$ must satisfy $\|M\|_{\ell_1^N \otimes_{\epsilon} \ell_1^N}\cdot\|M^{-1}\|_{\ell_\infty^N \otimes_\epsilon \ell_\infty^N} \geq 1$. If one wants to conclude $\|\cdot\|_c = \|\cdot\|_M$ from these results, one needs to investigate the equality case in the aforementioned inequality. We show that for any real and invertible matrix $M$, the following holds.
	
	\begin{proposition}
		For any real and invertible matrix $M$, we have:
		$$\|M^{-1}\|_{\ell^N_{\infty}\otimes_{\epsilon}\ell^N_{\infty}}\|M\|_{\ell^N_1\otimes_{\epsilon}\ell^N_1}\geq \frac{N}{\rho(\ell_{\infty}^N,\ell_{\infty}^N)}\geq \sqrt{\frac N 2}\geq 1.$$
		For $N \geq 3$, the last inequality above is strict.
	\end{proposition}
	
	The case $N=2$ needs to be treated separately. We show that for $N=2$ questions, the only games achieving equality are the CHSH game and variants thereof. We summarize this in the following theorem.
	
	\begin{theorem}
		The only invertible non-local game $M \in \mathcal M_N(\mathbb R)$ satisfying 
		$$\|M^{-1}\|_{\ell^N_{\infty}\otimes_{\epsilon}\ell^N_{\infty}}\|M\|_{\ell^N_1\otimes_{\epsilon}\ell^N_1} = 1$$
		have two questions ($N=2$) and are variants of the CHSH game: $M = a M_{\text{CHSH}}$ for some $a \neq 0$.
	\end{theorem}
	
	\section{Compatibility of quantum measurements}\label{sec:compatibility}

	This section contains the main definitions and results from the theory of quantum measurements, with the focus on (in-)compatibility and noisy measurements. 
	
	In Quantum Mechanics, a system is described by Hilbert space $\mathcal H$. Here, we shall consider only finite-dimensional Hilbert spaces: $\mathcal H \cong \C^d$, for a positive integer $d$, which corresponds to the number of degrees of freedom. For example, quantum bits (qubits) are described by the space $\mathbb C^2$. 
	Quantum states are formalized mathematically by \emph{density matrices}:
	$$\M^{1,+}_d := \{ \rho \in \M_d \, : \, \rho \geq 0 \text{ and } \Tr \rho = 1\},$$
	where $\mathcal M_d$ is the vector space of $d \times d$ complex matrices. Density matrices are positive semidefinite, a relation denoted by $\rho \geq 0$. 
	
	Let us now discuss measurements in Quantum Mechanics. Historically, quantum measurements were modelled by \emph{observables}: Hermitian operators acting on the system Hilbert space. The possible outcomes of the measurement are the eigenvalues of the observable, while the probabilities of occurrence are given by the celebrated \emph{Born rule}.  This formalism not only allows to obtain the probabilities of the different outcomes (via the Born rule), but also the post-measurement state of the quantum system (the \emph{wave function collapse}). In the current research, we are only concerned with the former, and thus we shall use the more general framework of Positive Operator Valued Measures (POVMs) \cite{nielsen00}. We shall write $[n]:= \{1, 2, \ldots, n\}$ for the set of the first $n$ positive integers.

	\begin{definition}
		A \emph{positive operator valued measure} (POVM) on $\M_d$ with $k$ outcomes is a $k$-tuple $A=(A_1, \ldots, A_k)$ of self-adjoint operators from $\M_d$ which are positive semidefinite and sum up to the identity:
		$$\forall i \in [k], \quad A_i \geq 0 \qquad \text{ and } \qquad \sum_{i=1}^k A_i = I_d.$$
		When measuring a quantum state $\rho$ with the apparatus described by $A$, we obtain a random outcome from the set $[k]$:
		$$\forall i \in [k], \qquad \mathbb P(\text{outcome} = i) = \Tr[\rho A_i].$$
	\end{definition}
	
	The vector of outcome probabilities $\left(\Tr[\rho A_i]\right)_{i=1}^k$ is indeed a probability vector; note that the properties of the operators $A_i$, called \emph{quantum effects}, are tailor made for this. This mathematical formalism used to described quantum measurements (or POVMs, or \emph{meters}) does not account for what happens with the quantum particle after the measurement. One can think that the particle is destroyed in the process of measurement (see Figure \ref{fig:measurement}) and thus only the outcome probabilities are relevant. 
	
	\begin{figure}[htb!]
		\centering
		\includegraphics{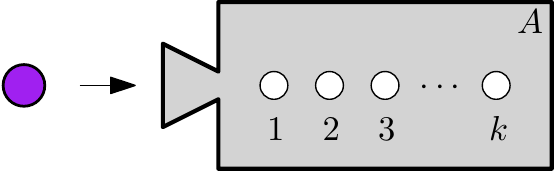}\qquad\qquad\qquad\qquad 
		\includegraphics{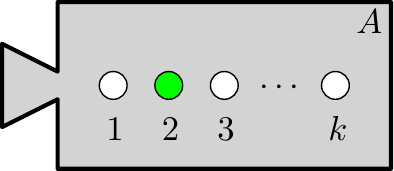}
		\caption{Diagrammatic representation of a quantum measurement apparatus. The device has an input canal and a set of $k$ LEDs which will turn on when the corresponding outcome is achieved. After the measurement is performed, the particle is destroyed, and the apparatus displays the classical outcome (here, $2$).}
		\label{fig:measurement}
	\end{figure}

	Several important classes of POVMs will be discussed in this paper: 
	\begin{itemize}
		\item \emph{von Neumann measurements}, where $A_i = \ketbra{a_i}{a_i}$, $i \in [d]$, for an orthonormal basis $\{\ket{a_i}\}_{i=1}^d$ of $\C^d$;
		\item \emph{trivial measurements}, where the matrices $A_i$ are scalar multiples of the identity: $A_i = p_i I_d$, for some probability vector $p = (p_1, p_2, \ldots, p_k)$. 
	\end{itemize}

	\bigskip
	
	Let us now define the notion of \emph{compatibility} for quantum measurements, which is central to this paper. Historically, in the physics literature, the notion of compatbility was closely related to that of commutativity of the quantum observables \cite{Heisenberg1927, Bohr1928}; indeed, sharp POVMs are compatible if and only if the corresponding observables commute. In the modern setting, suppose we want to measure two different physical quantities (modelled by two POVMs $A$ and $B$) on a given quantum particle in a state $\rho$. Having at our disposal just one copy of the particle, we cannot, in general, measure simultaneously $A$ and $B$. However, one can \emph{simulate} measuring $A$ and $B$ on $\rho$ with the help of a third POVM $C$, by \emph{classically} post-processing the output of $C$ to a pair of outcomes $(i,j)$ for $A$, respectively $B$, see Figure \ref{fig:compatibility}. Importantly, there are many pairs of POVMs $A$ and $B$ for which there is no such $C$, like the position and momentum operators of a particle in one dimension: it is impossible to attribute simultaneously an exact value to both position and momentum observables. 
	
	\begin{figure}[htb!]
		\centering
		\includegraphics[width=0.9\textwidth]{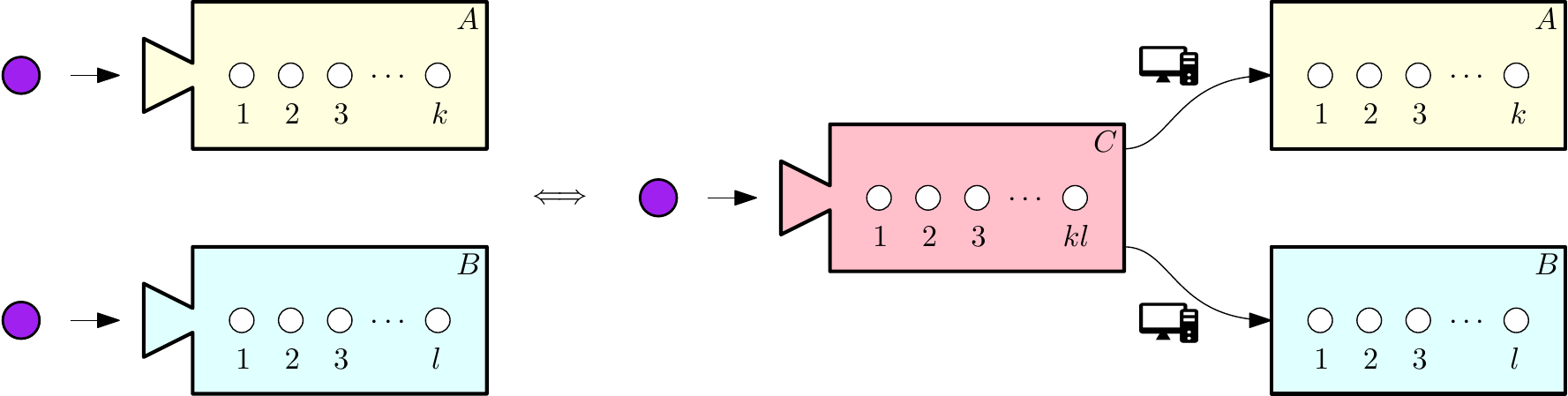}
		\caption{The joint measurement of $A$ and $B$ is simulated by by a third measurement $C$, followed by classical post-processing.}
		\label{fig:compatibility}
	\end{figure}
	
	Mathematically, we can either consider general post-processings or marginalization. We refer the reader to \cite{heinosaari2016invitation,heinosaari2022order} for more details. 
	
	\begin{definition}
		Two POVMs $A=(A_1, \ldots, A_k)$, $B=(B_1, \ldots, B_l)$ on $\M_d$ are called \emph{compatible} if there exists a \emph{joint POVM} $C=(C_{11}, \ldots, C_{kl})$ on $\M_d$ such that $A$ and $B$ are its respective \emph{marginals}:
		\begin{align*}
			\forall i \in [k], \qquad A_i &= \sum_{j=1}^l C_{ij}.\\
			\forall j \in [l], \qquad B_j &= \sum_{i=1}^k C_{ij}.
		\end{align*}
		
		More generally, a $g$-tuple of POVMs $\mathbf A = (A^{(1)}, \ldots, A^{(g)})$ is called compatible if there exists a POVM $C$ with outcome set $[k_1] \times \cdots \times [k_g]$ such that, for all $x \in [g]$, the POVM $A^{(x)}$ is the $x$-th marginal of $C$: 
		\begin{align*}
			\forall i_x \in [k_x], \qquad A^{(x)}_{i_x} &= \sum_{i_1 = 1}^{k_1} \cdots \sum_{i_{x-1} = 1}^{k_{x-1}}\sum_{i_{x+1} = 1}^{k_{x+1}} \cdots \sum_{i_g = 1}^{k_g} C_{i_1i_2 \cdots i_g}\\
			&= \sum_{\substack{\mathbf j \in [k_1] \times \cdots \times [k_g]\\ j_x = i_x}} C_{\mathbf j}.
		\end{align*}
	\end{definition}
	Note that the definition of compatibility given above can be formulated as a (feasibility) semi-definite program (SDP) \cite{boyd2004convex}. One can equivalently formulate the notion of compatibility with more general post-processings. 
	
	\begin{proposition}\label{prop:compatibility-postprocessing}
		An $N$-tuple of POVMs $\mathbf A = (A^{(1)}, \ldots, A^{(N)})$ is compatible if and only if there exists a joint POVM $(C_k)_{k \in [K]}$ and a family of conditional probabilities $\big(p_x(\cdot | \cdot)\big)_{x \in [N]}$ such that
		$$\forall x \in [N], \, \forall i \in [k_x], \qquad A^{(x)}_i = \sum_{k \in [K]} p_x(i | k) C_k.$$
	\end{proposition}

	We now consider the simplest possible setting, that of two 2-outcome POVMs $\{Q,I-Q\}$ and $\{P,I-P\}$, where $P,Q$ are $d \times d$ self-adjoint matrices satisfying $O \leq P, Q \leq I_d$. The pair of POVMs is compatible if and only if $\epsilon_0 \leq 1$ \cite{wolf2009measurements}, where 
	\begin{align}\label{equ:SDP}
		\epsilon_0:=\inf\big\{\epsilon\, : \,  \exists  \delta\geq 0 \quad \text{s.t.} \quad
		\delta + I-Q-P\geq 0,\,Q+\epsilon I - \delta \geq 0,\,P +\epsilon I - \delta \geq 0\big\},
	\end{align}
	where $\delta$ is a positive semidefinite matrix. The above formula corresponds to the value of a semidefinite program encoding the existence of a joint measurement for the POVMs $\{P, I-P\}$ and $\{Q, I-Q\}$. Generally, every SDP comes with a dual formulation. In our case the dual SDP is given below \cite{wolf2009measurements}:
	\begin{lemma}\label{lem:dual-SDP}
		Given the above optimization problem for deciding compatibility, its dual formulation is given by:
		$$\epsilon^*=\underset{X,Y,Z\geq0}{\sup}\bigg\{\Tr[X(Q+P-I)]-\Tr[YQ]-\Tr[PZ]\, \text{with}\, X\leq Y+Z,\, \Tr[Y+Z]=1\bigg\},$$
	\end{lemma}
	\begin{proof}
		Let us consider the following Lagrangian, corresponding the primal SDP ~\eqref{equ:SDP}.
		$$\mathcal{L}:=\epsilon-\langle X, \delta +I-Q-P\rangle-\langle Y, \epsilon I+Q-\delta\rangle-\langle Z, \epsilon I +P-\delta\rangle -\langle C,\delta\rangle$$
		Above $X,Y,Z,C$ are positive semidefinite matrices which represent the constraints of the primal optimisation problem. Due to the strict feasibility of the SDP we can calculate its dual optimal value is the same as the optimal one of the primal (Slater's condition, see \cite{boyd2004convex}). Thus, we have the following equality: $$\inf_{\epsilon,\delta} \sup_{X,Y,Z,C} \mathcal{L}=\sup_{X,Y,Z,C}\inf_{\epsilon,\delta}\mathcal{L}.$$
		A simple calculation shows that
		$$\underset{\epsilon,\delta}{\text{inf}}\,\mathcal{L}=\langle X,Q+P-I\rangle-\langle Y,Q\rangle-\langle P,Z\rangle$$
		with $\Tr[Y+Z]=1$ and $Z+Y-X-C=0\iff X\leq Y+Z$, which is precisely the dual formulation from the statement.
	\end{proof}
	
	In the following, we shall use the SDP value  to describe the compatibility threshold of the POVMs, that is the minimal quantity of noise that one needs to mix in, in order to render the POVMs compatible; such quantities go in the literature under the name of \emph{robustness of incompatibility}  \cite{designolle2019incompatibility}.
	
	\begin{definition}
		For a given parameter $\eta\in[0,1]$, given two POVMs $A=(A_1, \ldots, A_k)$, $B=(B_1, \ldots, B_l)$ on $\M_d$ one define their \emph{noisy} version as $A^{\eta}:=(A^{\eta}_1, \ldots, A^{\eta}_k)$, $B^{\eta}:=(B^{\eta}_1, \ldots, B^{\eta}_l)$ with
		\begin{align*}
			\forall i \in [k], \qquad A^{\eta}_i := \eta A_i + (1-\eta)\frac{I}{k}.\\
			\forall j \in [l], \qquad B^{\eta}_j := \eta B_j + (1-\eta)\frac{I}{l}.
		\end{align*}
		
	\end{definition}
	\begin{remark}
		In our simplified setting, we shall only consider POVMs with two outcomes $\mathcal{P}=\{P,\,I-P\}$, $\mathcal{Q}=\{Q,\,I-Q\}$ and their noisy versions. The definition given above can be rewritten as follows: 
		$$\mathcal{P}^{\eta}=\{P^{\eta},I-P^{\eta}\},$$ and $$\mathcal{Q}^{\eta}=\{Q^{\eta},I-Q^{\eta}\}.$$ The measurements $\mathcal{P}^{\eta}$ and $\mathcal{Q}^{\eta}$ can also be seen as convex mixtures of the measurement  $\mathcal{P}$ and $\mathcal{Q}$ with the trivial POVM $\mathcal{I}=(\frac{I}{2},\frac{I}{2})$. One has  $\mathcal{Q}^{\eta}=\eta \mathcal{Q}+(1-\eta)\mathcal{I}$ and $\mathcal{P}^{\eta}=\eta \mathcal{P}+(1-\eta)\mathcal{I}$.
	\end{remark}
	
	Let us now formalize the incompatibility robustness, in the symmetric case, where the same amount of white noise $(I/2, I/2)$ is mixed into the two POVMs; for the asymmetric version, see the incompatibility regions defined in \cite[Section III]{bluhm2018joint}.
	
	\begin{definition}\label{def:Gamma}
		For two (binary) measurements $\mathcal{P}$, $\mathcal Q$, we define their \emph{noise compatibility threshold} as:
		$$
		\Gamma(P,Q) :=\sup\big\{ \eta \in [0,1]\, : \,  \mathcal{P}^{\eta} , \mathcal{Q}^{\eta} \text{ are compatible }\big\}.$$
		
	\end{definition}
	
	\begin{proposition}\label{prop: noise compatibility threshold}
		The noise compatibility threshold for two (binary) measurements $\mathcal{P}$ and $\mathcal{Q}$ is given by: $$\Gamma(P,Q)=\frac{1}{1+2\epsilon^*},$$ where $\epsilon^*$ is the optimal value of the SDP from Lemma \ref{lem:dual-SDP}.
	\end{proposition}
	
	\begin{proof}
		Recalling that $\mathcal{P}^{\eta} , \mathcal{Q}^{\eta} \text{ are compatible }$
		is equivalent to $\exists\delta\geq 0$ with the following conditions:
		\begin{align*}
			\eta Q +(1-\eta)\frac{I}{2}-\delta&\geq0\\
			\eta P +(1-\eta)\frac{I}{2}-\delta&\geq0\\  \delta -\eta(P+Q-I)&\geq0
		\end{align*}
		Where it is easy to see that is equivalent to
		\begin{align*}
			Q +\epsilon I-\delta'&\geq0\\
			P +\epsilon I-\delta'&\geq0\\  
			\delta' -(P+Q-I)&\geq0
		\end{align*}
		with $\delta':=\frac{\delta}{\eta}$ and $\epsilon=\frac{1}{2}(\frac{1}{\eta}-1)\iff\eta=\frac{1}{2\epsilon + 1}$.
		By tacking the supremum over $\eta$ to compute the noise compatibility threshold $\Gamma(P,Q)$ with the following constraints:
		\begin{align*}
			Q +\epsilon I-\delta'&\geq0\\
			P +\epsilon I-\delta'&\geq0\\  
			\delta' -(P+Q-I)&\geq0
		\end{align*}
		is given by 
		\begin{align*}
			\Gamma(P,Q)&=\sup\Big\{\,\frac{1}{2\epsilon + 1}\Big|\,\exists \delta'\geq0\,
			,\,Q +\epsilon I-\delta'\geq0
			,\,P +\epsilon I-\delta'\geq0,\,             
			\delta' -(P+Q-I)\geq0\Big\}\\
			&=\frac{1}{2\inf\Big\{\,\epsilon\Big|\,P , Q \, \text{compatible}\Big\}+1}=\frac{1}{2\epsilon_0+1}=\frac{1}{2\epsilon^*+1}
		\end{align*}
		which ends the proof of the proposition.                 
		$$\Gamma(P,Q)=\sup\Big\{\, \eta\Big|\,\exists \delta\geq0
		,\, \eta Q +(1-\eta)\frac{I}{2}-\delta\geq0
		, \,\eta P +(1-\eta)\frac{I}{2}-\delta\geq0\\              
		,\,\delta -\eta(P+Q-I)\geq0\Big\}$$
		$$=\sup\Big\{\eta\Big|\exists\,\delta':=\frac{\delta}{\eta}\geq0, \,Q +(\frac{1}{\eta}-1)\frac{I}{2}-\delta'\geq0,\,P+(\frac{1}{\eta}-1)\frac{I}{2}-\delta'\geq0,\,             \delta' -(P+Q-I)\geq0\Big\}$$            
		
		With the following change of variable $$\epsilon=\frac{1}{2}(\frac{1}{\eta}-1) \iff \eta=\frac{1}{2\epsilon + 1}$$
		The compatibility threshold becomes
		\begin{align*}
			\Gamma(P,Q)&=\sup\Big\{\,\frac{1}{2\epsilon + 1}\Big|\,\exists \delta'\geq0\,
			,\,Q +\epsilon I-\delta'\geq0
			,\,P +\epsilon I-\delta'\geq0,\,             
			\delta' -(P+Q-I)\geq0\Big\}\\
			&=\frac{1}{2\inf\Big\{\,\epsilon\Big|\,P , Q \, \text{compatible}\Big\}+1}
		\end{align*}
		Thus, we have       $\Gamma(P,Q)=2/(2\epsilon^* +1)$, as announced. 
	\end{proof}
	\section{Tensor product of Banach spaces}\label{sec:Tensor Norms}
	In this section we wil give a brief overview of tensor norms with the aim of presenting Bell inequalities in the tensor norm framework. Tensor norms provide the natural mathematical framework for Bell inequalities, see the following survey \cite{palazuelos2016survey} and the reference therein. Let us start by recalling the projective and injective tensor norms for (finite-dimensional) Banach spaces. 
	
	\begin{definition}
		Given two finite-dimensional Banach spaces $X$ and $Y$ with their respective norms $\|\cdot\|_X$ and $\|\cdot\|_Y$, and $z\in X\otimes Y$, we define the \emph{projective tensor norm} of $z$ as:
		$$\|z\|_{X\otimes_{\pi} Y}:= \inf \left\{\sum_{i=1}^N\|x_i\|_X\|y_i\|_Y : z=\sum_{i=1}^N x_i\otimes y_i\right\},$$
		where the infimum is taken over all the decompositions of $z=\sum_{i=1}^N x_i\otimes y_i$ where N is a finite but arbitrary integer. 
		We write $X\otimes_{\pi}Y=(X\otimes Y,\|\cdot\|_{X\otimes_{\pi}Y})$, the Banach space induced by the projective tensor norm on $X\otimes Y$.  
	\end{definition}
	
	Every Banach space comes with a dual: 
	\begin{definition}
		Let $X$ a finite-dimensional Banach space. The space of all bounded linear functionals on $X$ is called its dual space and denoted by $X^*$. It comes equipped with a norm: 
		$$\forall \phi \in X^*, \qquad \|\phi\|_{X^*}:= \sup_{\|x\|_X\leq 1} |\phi(x)|.$$
	\end{definition}
	
	We now introduce the other tensor norm of importance to us. 
	
	\begin{definition}\label{def: injective norm}
		Given two finite-dimensional Banach spaces $X$ and $Y$ with their respective norms $\|\cdot\|_X$ and $\|\cdot\|_Y$, and $z\in X\otimes Y$, we define the \emph{injective tensor norm} of $z$ as:
		$$\|z\|_{X\otimes_{\epsilon}Y}:= \underset{\alpha \in \mathbb{B}(X^*), \beta\in \mathbb {B}(Y^*)}{\sup}|\langle z,\alpha\otimes \beta\rangle|,$$
		where $\mathbb{B}(X^*)$ and $\mathbb{B}(Y^*)$ are the unit balls of $X^*$ and $Y^*$. 
		
		We write $X\otimes_{\epsilon}Y=(X\otimes Y,\|\cdot\|_{X\otimes_{\epsilon}Y})$, the Banach space induced by the injective norm on $X\otimes Y$.
	\end{definition}
	
	It is known that the projective and the injective tensor product play the role of maximal and the minimal norm respectively that we can put naturally in the algebraic tensor product, for that we give the following definition of a reasonable crossnorm.  
	\begin{definition}\label{def: reasonnable norm}
		Let $z\in X\otimes Y$, we say that a norm $\alpha$ on $X\otimes Y$ given by $\|z\|_{X\otimes_{\alpha}Y}$ is a \emph{reasonable crossnorm} (or a \emph{tensor norm}) if for $z=x\otimes y$ we have:
		$$\|z\|_{X\otimes_{\alpha}Y}\leq \|x\|_X\|y\|_Y$$
		and the dual $\phi=\phi_1\otimes\phi_2\in X^*\otimes Y^*$ satisfies $$ \|\phi\|_{X^*\otimes_{\alpha}Y^*}\leq \|\phi_1\|_{X^*}\|\phi_2\|_{Y^*}.$$
		We write $X\otimes_{\alpha}Y=(X\otimes Y,\|\cdot\|_{X\otimes_{\alpha}Y})$, the Banach space induced by $\alpha$ on $X\otimes Y$.
	\end{definition}
	The definition above can be found in \cite[page 127]{ryan2002introduction}, with the following equivalent statement. 
	\begin{proposition}\cite[Proposition 6.1]{ryan2002introduction}\label{Prop: reasonnable norm}
		Consider two finite-dimensional Banach spaces $X$ and $Y$ with their respective norms $\|\cdot\|_X$ and $\|\cdot\|_Y$. A norm $\alpha$ on $X \otimes Y$ is a reasonable crossnorm if and only if for all $z\in X\otimes Y$, we have
		$$\|z\|_{X\otimes_{\epsilon} Y}\leq \|z\|_{X\otimes_\alpha Y}\leq \|z\|_{X\otimes_\pi Y}.$$
	\end{proposition}
	\begin{remark}
		The injective and the projective tensor norm are dual to each other, in the following sense \cite{ryan2002introduction}: 
		
		\begin{align*}
			\|z\|_{X\otimes_{\pi}Y}:=\underset{\|\alpha\|_{X^*\otimes_{\epsilon}Y^*}\leq 1}{\sup}\langle\alpha,z\rangle,\\
			\|z\|_{X\otimes_{\epsilon}Y}:=\underset{\|\alpha\|_{X^*\otimes_{\pi}Y^*}\leq 1}{\sup}\langle\alpha,z\rangle.
		\end{align*}
	\end{remark}
	In general, for each tensor norm $\|\cdot\|_{X\otimes_{\alpha}Y}$ we can define its \emph{dual tensor norm} we note it by $\alpha^*$ and we have the following definition. 
	\begin{definition}\label{def: dual tensor norm}
		Consider two finite dimensional Banach spaces $X$ and $Y$ with their respective norms $\|\cdot\|_X$ and $\|\cdot\|_Y$. Let $z\in X\otimes Y$ and $v\in X^*\otimes Y^*$, the dual tensor norm $\alpha^*$ of a given tensor norm $\alpha$ is defined by $$\|z\|_{X\otimes_{\alpha^*}Y}:=\sup\{|\langle v,z\rangle|;\,\|v\|_{X^*\otimes_{\alpha}Y^*}\leq 1\}.$$
		We write $X\otimes_{\alpha^*}Y=(X\otimes Y,\|\cdot\|_{X\otimes_{\alpha^*}Y})$, the Banach space induced by the norm $\alpha^*$ on $X\otimes Y$.
	\end{definition}
	\begin{remark}
		With the definition above, we have the nice identification between the dual space of the tensor product of two spaces endowed with a tensor norm $\alpha$ and the dual space of each of the two spaces endowed with the dual norm $\alpha^*$ where we have $$(X\otimes_{\alpha}Y)^*=X^*\otimes_{\alpha^*}Y^*.$$ 
	\end{remark}
	
	One last definition we want to recall that will play a fundamental role for Bell inequalities which is the reasonable norm known as $\gamma_2$ norm.

	\begin{definition}
		Given two finite-dimensional Banach spaces $X$ and $Y$ with their respective norms $\|\cdot\|_X$ and $\|\cdot\|_Y$, define the tensor norm $\gamma_2$ of $z\in X\otimes Y$ by:
		$$\|z\|_{X\otimes_{\gamma_2}Y}:= \inf\left\{\underset{\alpha^*\in \mathbb{B}(X^*)}{\sup}\left(\sum_{i=1}^N|\alpha^*(x_i)|^2\right)^{\frac{1}{2}}\underset{\beta^*\in \mathbb{B}(Y^*)}{\sup}\left(\sum_{j=1}^N|\beta^*(y_j)|^2\right)^{\frac{1}{2}}
		:z=\sum_{i=1}^N x_i\otimes y_i\right\},$$
		where the infimum is taken over all decompositions of $z=\sum_{i=1}^N x_i\otimes y_i$ with $x_i\in X$ and $y_j\in Y$.
	\end{definition}
	
	\section{Bell inequalities and non-local games}\label{sec:non-locality}
	In this section we introduce the notion of \emph{non-local games} and show how it incorporates the non-local properties of the quantum world via the well-known Bell inequalities. We then recall how the natural framework for these games is the metric theory of tensor products of finite-dimensional Banach spaces. 
	
	The non-local aspect of quantum mechanics can be incorporated in the non-local game framework. In this paper, we only consider non-local games with two players, \emph{Alice} and \emph{Bob}, which are collaborating to win the game. During the game, a third person known as the \emph{Referee}, will ask a certain number of question to the players which are not allowed to communicate. The two protagonists reply cooperatively to the referee with some answers. The referee will decide to accept or reject the answers, declaring a win or a loss. Note that in this paper we are going to consider games with arbitrary number of questions $N$, but only with two answers ($+1$ or $-1$); in the general case, Alice and Bob can give answers from a fixed set of given cardinality. 
	
	During the game, players have access to a predetermined set of resources: this determines the type of strategy they are permitted to use.  In this work, we shall consider \emph{classical strategies} and \emph{quantum strategies}. 
	
	In classical strategies, the players share samples from a classical random variable that they can use to produce their answers locally. When using quantum strategies, the players share a bipartite quantum state, on which they can act locally with transformations and measurements. 
	
	We shall focus on \emph{correlation games}, where the payoff of the game depends on the correlation of the $\pm 1$ answers $ab$, weighted by real numbers $M_{xy}$ depending on the question: 
	$$\text{payoff} = \sum_{x,y \in [N]} \sum_{a,b \in \{\pm 1\}} M_{xy}\,ab \cdot \mathbb P(a,b|x,y),$$
	where $\mathbb P(a,b|x,y)$ is the (strategy-dependent) probability that Alice and Bob answer respectively $a$ and $b$, when presented with the questions $x,y \in [N]$. The matrix $M \in \mathcal M_N(\mathbb R)$ encodes the rules of the game, and it is called the \emph{Bell functional} \cite{palazuelos2016survey}. In the following, we discuss the optimal classical and quantum strategies for a given non-local game $M$.
	
	\begin{definition}
		The \emph{classical bias of the game} $M$ is defined as the optimisation problem 
		$$\beta(M):=\sup \Big|\sum_{x,y=1}^N\,\sum_{a,b\in\{\pm1\}}\, M_{xy}\, ab\, \mathbb P_c(a,b|x,y)\Big| $$
		where the supremum is taken over all \emph{classical strategies}
		$$\mathbb P_c(a,b|x,y)=\int_{\Lambda} \mathbb P_A(a|x,\lambda)\,\mathbb P_B(b|y,\lambda)\, \mathrm d\mu(\lambda).$$
		Above, $\mathbb P_A$, resp.~$\mathbb P_B$ correspond to Alice's, resp.~Bob's strategies, which can depend on the shared random variable $\lambda$ having distribution $\mu$.
	\end{definition}
	
	Introducing the expectation values with respect to the outputs $a,b$
	$$A_x(\lambda):=\sum_{a\in\{\pm 1\}} a\, \mathbb P_A(a|x,\lambda) \qquad \text{ and } \qquad  B_y(\lambda):=\sum_{b\in\{\pm1\}} b\, \mathbb P_B(b|y,\lambda),$$
	we have 
	\begin{align*}
		\beta(M)&=\sup_{\mathbb P_A, \mathbb P_B, \mu}\Big|\sum_{x,y=1}^N\,\sum_{a,b\in\{\pm1\}}\, M_{xy}\,a\,b\, \int_{\Lambda} \mathbb P_A(a|x,\lambda)\, \mathbb P_B(b|y,\lambda)\,\mathrm d\mu(\lambda)\Big|\\
		&=\sup_{A_x, 
			B_y, \mu}\Big|\sum_{x,y=1}^N\, M_{xy}\, \int_{\Lambda} A_x(\lambda)\,B_y(\lambda)\,d\mu(\lambda)\Big|\\
		&=\sup_{\gamma}\Big|\sum_{x,y=1}^N\, M_{xy}\,\gamma_{x,y}\Big|,
	\end{align*}
	where the matrix $\gamma=  (\gamma_{x,y})$ is a classical correlation matrix, containing the relevant information from the set of classical strategies.
	\begin{definition}
		We define the set of \emph{classical correlations} as 
		$$\mathbb L:=\left\{\gamma_{x,y}\,\Big|\,\gamma_{x,y}=\int_{\Lambda} A_x(\lambda) B_y(\lambda) \, \mathrm d\mu(\lambda); \, |A_x(\lambda)|,|B_y(\lambda)|\leq 1  \right\} \subseteq \mathcal{M}_N(\mathbb{R})$$
		where $\lambda$ is a random variable shared by Alice and Bob, following a probability distribution $\mu$.
	\end{definition}
	
	Using the definition above, the maximum payoff of a game $M$, using classical strategies, can be understood as the maximum overlap of the Bell functional $M$ defining the game with the set of classical correlations. 
	\begin{proposition}
		The classical bias of the game defined by a Bell functional $M$ is:
		$$\beta(M)=\sup_{\gamma \in \mathbb L}\Bigg \{\Big| \sum_{x,y=1}^N \,M_{xy}\,\gamma_{x,y}\Big| \Bigg\}.$$    
	\end{proposition}

	\medskip
	
	We now move on to the quantum setting, where the players are allowed to use quantum strategies, that is they are allowed to perform local operations on a shared entangled state. 
	
	\begin{definition}
		The \emph{quantum bias of the game} $M$ is defined as the optimisation problem 
		$$\beta^*(M):=\sup \Big|\sum_{x,y=1}^N\,\sum_{a,b\in\{\pm1\}}\, M_{xy}\, a\,b\, \mathbb P_q(a,b|x,y)\Big| $$
		where the supremum is taken over all \emph{quantum strategies}
		$$\mathbb P_q(a,b|x,y)=\Tr\Big[\rho\, (A_{a|x}\otimes B_{b|y})\Big],$$ 
		where $\rho$ is a bipartite shared quantum state (of arbitrary dimension), and, for all questions $x,y$, $(A_{\pm|x})$, resp.~$(B_{\pm|y})$ are POVMs on Alice's, resp.~Bob's quantum system.
	\end{definition}
	
	Introducing the operators 
	$$A_x := \sum_{a \in \{\pm 1\}} a\,A_{a|x} \qquad \text{ and } \qquad B_y := \sum_{b \in \{\pm 1\}} b\,B_{b|y},$$
	and performing a similar computation as in the case of classical strategies, we are led to following definition and expression for the quantum bias of a non-local correlation game $M$. 
	
	\begin{definition}
		We define the set of \emph{quantum correlations} as
		$$\mathbb Q:=\left\{\gamma_{x,y}\,\Big|\,\gamma_{x,y}=\Tr\bigg[\rho \cdot (A_x\otimes B_y) \bigg];\, \|A_x\|_{\infty},\|B_y\|_{\infty}\leq 1\right\}\subseteq \mathcal{M}_N(\mathbb{R}).$$
		Above, $\rho$ is a bipartite quantum state of arbitrary dimension, and $A_x, B_y$ are observables of norm less than one. 
	\end{definition}
	
	\begin{proposition}
		The quantum bias of the game defined by a Bell functional $M$ is:
		$$\beta^*(M)=\sup_{\gamma \in \mathbb Q}\Bigg \{\Big| \sum_{x,y=1}^N \,M_{xy}\,\gamma_{x,y}\Big| \Bigg\}.$$    
	\end{proposition}
	
	\begin{remark}
		The correlation games discussed above are also known in the literature as \emph{XOR games}, when the set of outputs is $\{0,1\}$ (instead of $\{\pm 1\}$) \cite{regev2015quantum}.
	\end{remark}
	
	Since classical correlations are a subset of the quantum correlations (corresponding to diagonal operators $A_x, B_y$), the quantum bias of the game must be always larger or equal the classical bias. In some cases, the quantum bias $\beta^*(M)$ is strictly larger than the classical one, which can be understood physically as the existence of quantum correlations can not  be reproduced within a classical local hidden variable model. This motivates the following definition.
	
	\begin{definition}
		For a given non-local game described by a \emph{Bell functional} $M$, we say that we have a \emph{Bell violation} if $\beta^*(M)>\beta(M)$.
	\end{definition}
	
	Now we will recall the results on the profound link between the classical bias of a game and its quantum bias within their respective tensor norm description.
	
	\begin{theorem}\cite{palazuelos2016survey}
		Consider a non-local correlation game characterized by the matrix  $M \in \mathcal M_N(\mathbb R)$. 
		\begin{itemize}
			\item The classical bias of the game is equal to the injective tensor norm of $M$:
			$$\beta(M)=\|M\|_{\ell^N_1(\mathbb R)\otimes_{\epsilon}\ell^N_1(\mathbb R)}.$$
			\item The quantum bias of the game is equal to the $\gamma^*_2$ tensor norm of $M$:
			$$\beta^*(M)=\|M\|_{\ell^N_{1}(\mathbb{R})\otimes_{\gamma_2^*} \ell^N_{1}(\mathbb{R})}.$$
		\end{itemize}
		where we recall from the Definition \ref{def: dual tensor norm} that $$\|M\|_{\ell^N_{1}(\mathbb{R})\otimes_{\gamma_2^*} \ell^N_{1}(\mathbb{R})}:=\sup\Big\{|\langle v, M\rangle|\,;\, \|v\|_{\ell^N_{\infty}(\mathbb{R})\otimes_{\gamma_2} \ell^N_{\infty}(\mathbb{R})}\leq 1\Big\}.$$
	\end{theorem}
	Tsirelson showed in \cite{tsirel1987quantum} the following theorem that links the classical and the quantum bias of an XOR game with  famous Grothendieck constant $K_{G}^{\R}$ that plays a fundamental role in the theory of tensor product of Banach spaces, see also \cite[Corollary 3.3]{palazuelos2016survey}.
	\begin{theorem}
		Consider a non-local correlation game characterized by the matrix  $M \in \mathcal M_N(\mathbb R)$.
		$$\beta^*(M)\leq K_{G}^{\R}\,\beta(M).$$
	\end{theorem}

	From the result above, one can easily see that Bell inequality violations ($\beta^*(M) > \beta(M)$) can be understood as tensor norm ratios. The result above shows the intrinsic link between Bell inequality violation and tensor norms, which motivates our framework on using tensor norms.
	
	\bigskip
	
	Let us now discuss the CHSH non-local game \cite{clauser1969proposed}. 
	
	\begin{definition}
		The CHSH game is given by the particular \emph{Bell functional} defined as the following:
		$$M_{\text{CHSH}} = \frac 1 2 \begin{bmatrix}
			1 & 1 \\ 
			1 & -1 \end{bmatrix},$$
		
	\end{definition}
	In this subsection, we recall the result of \cite{wolf2009measurements} where they made the link between the maximal violation of Bell inequality and the compatibility of the quantum measurement in the CHSH game.
	Precisely the maximal violation of the CHSH inequality is equivalent to the dual of formulation of the compatibility problem as an SDP \cite{wolf2009measurements}.
	
	\begin{theorem}
		Two dichotomic measurements $A=(A_0,A_1)$, $B=(B_0,B_1)$ are incompatible if and only if they enable violation of the CHSH inequality. More precisely, the optimal value of the CHSH inequality is related to 
		$$\underset{\psi,B_0,B_1}{\sup}\bra{\psi}\mathbb{B}\ket{\psi}=\frac{1}{\Gamma(A)}.$$
		with $$\mathbb{B}:=\sum_{x,y=0}^1 M_{\textrm{CHSH}}(x,y)A_x\otimes B_y=\frac{1}{2}(A_0\otimes B_0+A_0\otimes B_1+A_1\otimes B_0-A_1\otimes B_1).$$ and $\Gamma(A)$ is the noise compatibility threshold. 
	\end{theorem}
	\begin{proof}
		The proof of this theorem is basically the following remark, where in \cite{wolf2009measurements} they have showed that $$\underset{\psi,B_0,B_1}{\sup}\bra{\psi}\mathbb{B}\ket{\psi}=1+2\epsilon^*.$$ By combining this result and the proposition \ref{prop: noise compatibility threshold} where Alice measurement apparatus are described by the POVMs $\{Q,I-Q\}$ and $\{P,I-P\}$ which ends the proof.
	\end{proof}
	
	\section{The tensor norm associated to a game}\label{sec:non-locality-norm}
	In this section we will introduce the notion of non-locality using our framework of tensor norms. For that we consider as the previous section a fixed \emph{quantum game}, and for \emph{fixed Alice measurements} we introduce the \emph{$M$-Bell-(non)locality} notion. This quantity will characterise all the observed non-local effects in Alice side. To do so she will calculate the following tensor norm $\|A\|_M$ given by her fixed measurement apparatus. This quantity is obtained by optimizing over all the shared quantum state and all Bob measurement apparatus. We say that Alice measurement apparatus are \emph{$M$-Bell-local} if such $\|A\|_M$ is \emph{less than or equal the classical bias of the game}  and if is not we say that her measurement are \emph{$M$-Bell non-local}.\\
	
	The physical motivation of such statement can be understood as the following, for a fixed quantum game no matter the optimisation over all the shared quantum states and Bob measurement we cannot do better than the classical bias of the game, which means that one cannot do better than the classical setting even if we use the quantum strategies.\\
	For that we will give the precise definition of $\|A\|_M$ and the \emph{$M$-Bell-(non)locality} notion. We will show the main theorem of this section that $\|A\|_M$ is a \emph{tensor norm} in $(\mathbb{R}^N,\|\cdot\|_M)\otimes (\mathcal{M}_d,\|\cdot\|_{\infty})$ for a fixed invertible quantum game $M$.\\
	
	As a starting point, we give the two main definitions of this section. 
	\begin{definition}\label{def:A-M}
		Consider a fixed $N$-input, 2-outcome non-local game $M \in \mathcal M_N(\mathbb R)$.
		Fix also Alice's measurements, a $N$-tuple of binary observables $A = (A_1, \ldots, A_N) \in \mathcal M^{sa}_d(\mathbb C)^N$. The largest quantum bias of the game $M$, with Alice using the observable $A_x$ to answer question $x \in [N]$, is given by
		$$\sup_{\|\psi\| = 1} \sup_{\|B_y\| \leq 1} \Big \langle \psi \Big | \sum_{x,y = 1}^N\,M_{xy} \,A_x \otimes B_y \Big | \psi \Big \rangle = \underset{\|B_y\|\leq 1}{\sup}\lambda_{\max}\left[\sum_{x,y}^N\,M_{xy}\,A_x\otimes B_y\right]=: \|A\|_M,$$
		where the suprema are taken over bipartite pure states $\psi \in \mathbb C^d \otimes \mathbb C^D$ and over Bob's observables $B = (B_1, \ldots, B_N) \in \mathcal M^{sa}_D(\mathbb C)^N$, where $D$ is a free dimension parameter. We shall later show in Theorem \ref{thm:M-tensor-norm} that this quantity defines a (tensor) norm.
	\end{definition}
	
	\begin{remark}
		In the definition above, the dimension of Alice's measurements is fixed ($d$), while the dimension of Bob's Hilbert space ($D$) is free. In the following we will show that one can assume, without loss of generality, that Alice and Bob have Hilbert spaces of the same dimension ($D=d$ suffices in the optimization problem).
		
		Let us consider $D\geq d$, a quantum state $\ket \psi \in \mathbb C^d \otimes \C^D$, and $N$ binary measurement operators $B_1, \ldots, B_N \in \mathcal M^{sa}_D(\mathbb C)$. The idea is that the Schmidt decomposition of the bipartite pure quantum state $\ket \psi$ will induce a reduction of the effective dimension of Bob's Hilbert space from $D$ to $d$. We start from the Schmidt decomposition of $\ket \psi$
		$$\ket \psi=\sum_{i=1}^d \sqrt{\lambda_i}\ket{a_i}\otimes\ket{b_i}.$$
		Note that in the equation above, the number of terms is bounded by the smallest of the two dimensions, that is $d$. The orthonormal family $\{\ket{b_i}\}_{i \in [d]}$ spans a subspace of dimension $d$ inside $\mathbb C^D$. Consider an arbitrary orthonormal \emph{basis} $\{\ket{\tilde b_i}\}_{i \in [d]}$ of $\mathbb C^d$ and the isometry
		$$V : \mathbb C^d \to \mathbb C^D \quad \text{ such that } \quad  \forall i \in [d], \quad V\ket{\tilde b_i} = \ket{b_i}.$$
		Let us now introduce the quantum state
		$$\mathbb C^d \otimes \mathbb C^d \ni \ket{\tilde \psi}:= \sum_{i=1}^d \sqrt{\lambda_i}\ket{a_i}\otimes\ket{\tilde b_i}$$
		and the measurement operators 
		$$ \mathcal M^{sa}_d(\mathbb C) \ni \tilde B_y := V^* B_y V, \quad \forall y \in [N].$$
		The normalization of the state and the fact that the $\tilde B_y$ are contractions follow from the isometry property of the operator $V$. We now have 
		\begin{align*}
			\Big \langle \psi \Big | \sum_{x,y = 1}^N M_{xy} A_x \otimes B_y \Big | \psi \Big \rangle &= \sum_{x,y = 1}^N M_{xy}\sum_{i,j=1}^d\sqrt{\lambda_i\lambda_j} \langle a_i | A_x | a_j \rangle \underbrace{\langle b_i | B_y | b_j \rangle}_{=\langle \tilde b_i | V^* B_y V | \tilde b_j \rangle}\\
			&= \sum_{x,y = 1}^N M_{xy}\sum_{i,j=1}^d\sqrt{\lambda_i\lambda_j} \langle a_i | A_x | a_j \rangle \langle \tilde b_i | \tilde B_y | \tilde b_j \rangle\\
			&= \Big \langle \tilde\psi \Big | \sum_{x,y = 1}^N M_{xy} A_x \otimes \tilde B_y \Big | \tilde\psi \Big \rangle.
		\end{align*}
		The above computation shows that any correlation that can be obtained with Bob's Hilbert space of dimension $D$ can also be obtain with a Hilbert space of dimension $d$, equal to that of Alice. 
	\end{remark}
	\begin{definition}\label{def:Bell-local}
		Given a non-local game $M$, we say that Alice's measurements $A = (A_1, \ldots, A_N)$ are \emph{$M$-Bell-local} if for any choice of Bob's observables $B$ and for any shared state $\psi$, one cannot violate the Bell inequality corresponding to $M$: 
		$$\|A\|_M\leq \beta(M).$$ 
		If this is not the case, we call Alice's measurements \emph{$M$-Bell-non-local}.
	\end{definition}
	
	Instead of using definition \ref{def:A-M} we will use another simple equivalent formulation of $\|A\|_M$. To do so, we will consider $\|A\|_M$ as an optimization problem using an SDP, and we will give its equivalent formulation as a dual of the primal SDP.
	
	\begin{lemma}
		Given a quantum game $(M_{xy})_{\{x,y=1\}}^N$ we can characterise the following equivalent formulation of $\|A\|_M$ : 
		$$\|A\|_{M}=\lambda_{\max}\left[\sum_{y=1}^N \bigg| \sum_{x=1}^N M_{xy}\,A_x\bigg|\right].$$
	\end{lemma}
	
	\begin{proof}
		Remark that the definition above is equivalent to
		$$\|A\|_M=\sup_{\|\psi\| = 1}\quad\sup_{\|B_y\| \leq 1} \Big \langle \psi \Big | \sum_{x,y = 1}^N\,M_{xy}\,A_x \otimes B_y \Big | \psi \Big \rangle$$
		with $\ket{\psi}=\sqrt{\rho}\otimes I \sum_{i=1}^d\ket{ii}$ and $\rho$ is a density matrix (this is the classical purification trick in quantum information theory, expressing any bipartite pure state as a local perturbation of the maximally entangled state $\sum_i \ket{ii}$). Then, one has 
		$$\|A\|_M=\sup_{\rho\geq 0\,;\,\Tr\rho=1}\quad \sup_{\|B_y\| \leq 1}\bigg\{\sum_{x,y=1}^N\,M_{xy}\Tr\Big[\sqrt{\rho}
		\, A_x\, \sqrt{\rho}\, B_y^{\top}\Big]\bigg\}$$
		
		Using the following variable change $B_y^{\top}=2S_y-I$ for all $y$, we have
		
		$$\|A\|_M=\underset{S_y\,,\,\rho}{\sup}\bigg\{\sum_{y=1}^N\Tr\Big[\sqrt{\rho}\, A'_y\, \sqrt{\rho}\, (2S_y-I)\Big]\bigg\}$$
		where $A'_y=\sum_{x=1}^N M_{xy}\,A_x$ and the optimisation problem is on $S_y$ for all $y$ and $\rho$ with the following constraint : $0\leq S_y\leq I$ and $\rho\geq 0$ and $\Tr\rho=1$. 
		
		Then $$\|A\|_M=\underset{S_y\,,\,\rho}{\sup}\bigg\{\sum_{y=1}^N\Tr\Big[A'_y(2\,\sqrt{\rho}\,S_y\,\sqrt{\rho}-\rho)\Big]:\, 0\leq S_y\leq I\, ,\, \rho\geq0\,, \,\Tr\rho=1\bigg\}$$
		
		We consider now one last change of variable $S'_y=\sqrt{\rho}\,S_y\,\sqrt{\rho}$.
		
		$$\|A\|_M=\underset{S'_y\,,\,\rho}{\sup}\bigg\{\sum_{y=1}^N\Tr\Big[A'_y(2S'_y-\rho)\Big]:\,0\leq S'_y\leq \rho\,,\,\rho\geq 0\,,\,\Tr\rho=1\bigg\}$$
		
		where the optimisation problem now is on $S'_y$ for all $y$ and $\rho$ with the constrained above.\\
		
		We will formulate $\|A\|_M$ as an SDP and we will compute its dual. 
		
		For that, we consider the following Lagrangian :
		$$\mathcal{L}=\sum_{y=1}^N\, \Tr\Big[\,A'_y\,(2S'_y-\rho)\Big]\,+\,\langle X,\rho\rangle\, +\,\sum_{y=1}^N\, \langle X_y,S'_y\rangle\, +\, \epsilon(1-\Tr \rho)+\sum_{y=1}^N\, \langle \rho-S'_y,Z_y\rangle.$$
		with $X$, $X_y$, $\epsilon$, $Z_y$ are the constraints respectively for $\rho\geq0$, $S'_y\geq0$, $\Tr\rho=1$ and $S'_y\leq \rho$.\\
		Then by using the SDP duality one has
		$$\|A\|_M=\sup_{S'_y\,,\,\rho}\quad \inf_{X\,,\,X_y\,,\epsilon\,,Z_y} \quad \mathcal
		L=\inf_{X\,,\,X_y\,,\epsilon\,,Z_y} \quad  \sup_{S'_y\,,\,\rho} \quad \mathcal L$$
		
		with the following constraints:
		\begin{itemize}
			\item $X\geq 0$ and $\forall y$ $X_y\,,Z_y\geq 0$ are positive semidefinite matrices
			\item $\epsilon \in \mathbb R$ is unconstrained
		\end{itemize}
		Using the duality given above and by tacking the suprema first over $S'_y$ and $\rho$, we have:  $$\sup_{S'_y\,,\rho}\,\mathcal L=\begin{cases}\,\epsilon\,,\qquad\forall y\,,\, 2A'_y+X_y-Z_y=0\quad\text{and}\quad X-\sum_{y=1}^N A'_y-\epsilon I +\sum_{y=1}^N Z_y=0.\\
			\,+\infty
		\end{cases}$$ 
		
		Now by tacking the infimum over the constraints
		\begin{align*}
			\|A\|_M&=\underset{X\,,\,X_y\,,\,\epsilon\,,\,Z_y}{\inf}\bigg\{\epsilon\Big|\quad \forall y\,,\, 2A'_y+X_y-Z_y=0\,;\,
			X-\sum_{y=1}^N A'_y-\epsilon I +\sum_{y=1}^N Z_y=0\,\bigg\}\\
			&=\underset{X_y\,,\,\epsilon}{\inf}\bigg\{\epsilon\Big|
			\quad\forall y \,,\,2A'_y\leq Z_y\quad\text{and}\quad \sum_{y=1}^N A'_y+\epsilon I \geq \sum_{y=1}^N Z_y\,\bigg\}
		\end{align*}
		where in the last equality we have used that $X\geq 0$, $ Z_y\geq 0$. With the constraints on $Z_y$ and $Z_y\geq 2A'_y$, we can choose $Z_y:=2(A')_y^+$ with $(A')_y^+$ is the positive part of $A'_y=(A')_y^+ - (A')_y^-$; this is the smallest (with respect to the positive semidefinite order) choice for $Z_y.$ Using the optimal value above
		$$\|A\|_M=\underset{\epsilon}{\inf}\bigg\{\epsilon\Big|\,\epsilon I\geq \sum_{y=1}^N (A')_y^+ + (A')_y^-\bigg\}$$
		Then 
		$$\|A\|_{M}=\lambda_{\max}\left[\sum_{y=1}^N \bigg| \sum_{x=1}^N\,M_{xy}\,A_x\bigg|\right]$$
		
	\end{proof}
	In the following, we shall exploit the new formulation of $\|A\|_M$ and we will show in the lemma below that $\|\cdot\|_M$ is a norm for any \emph{invertible} game $M$.
	\begin{lemma}\label{lem:norm-M}
		Given an \emph{invertible} game $M$, the $M$-Bell-locality quantity $\|A\|_M$ verifies the following two properties:
		$$\|A\|_M\geq 0,$$ 
		$$\|A+A'\|_M\leq\|A\|_M+\|A'\|_M.$$
		In particular, $\|\cdot\|_M$ is a norm.
	\end{lemma}
	
	\begin{proof}
		To prove the first property we shall prove that :
		\begin{itemize}
			\item $\forall \alpha \in \mathbb{R}\, \text{we have}\, \|\alpha A\|_M=|\alpha|\|A\|_M.$
			\item $\|A\|_M=0 \implies A=0$
		\end{itemize}
		Obviously we have $$\|\alpha A\|_M=\lambda_{\max}\left[\sum_{y=1}^N \bigg| \sum_{x=1}^N \alpha M_{xy}\,A_x\bigg|\right]=|\alpha|\lambda_{\max}\left[\sum_{y=1}^N \bigg| \sum_{x=1}^N M_{xy}\,A_x\bigg|\right]=|\alpha|\|A\|_M$$
		For the second property we have $$\|A\|_M=\lambda_{\max}\left[\sum_{y=1}^N \bigg| \sum_{x=1}^N  \,M_{xy}\,A_x\bigg|\right]=0 \implies \sum_{x=1}^N  \, M_{xy}\, A_x=0\implies A=0$$
		where the first implication is due to the positivity of $|M^{\top}\,A|$ and the second implication is obtained by the assumption of the invertibility of $M$ and thus $M^{\top}$.\\
		For the second property we have
		\begin{align*}
			\|A+A'\|_M&=\sup_{\|\psi\| = 1}\quad\sup_{\|B_y\| \leq 1} \Big \langle \psi \Big | \sum_{x,y = 1}^N\,M_{xy}\,(A_x+A'_x) \otimes B_y \Big | \psi \Big \rangle\\
			&\leq\sup_{\|\psi\| = 1}\quad\sup_{\|B_y\| \leq 1} \Big \langle \psi \Big | \sum_{x,y = 1}^N\,M_{xy}\,A_x \otimes B_y \Big | \psi \Big \rangle+ \sup_{\|\psi\| = 1}\quad\sup_{\|B_y\| \leq 1} \Big \langle \psi \Big | \sum_{x,y = 1}^N\,M_{xy}\,A'_x \otimes B_y \Big | \psi \Big \rangle
		\end{align*}
		Hence we have $$\|A+A'\|_M\leq \|A\|_M+\|A'\|_M.$$
	\end{proof}
	
	In the last lemma we have shown that $\|A\|_M$ is a norm (for an \emph{invertible Bell functional} $M$). We shall call this norm the \emph{$M$-Bell-locality norm}. 
	
	The (real) vector spaces $\mathbb R^N$, resp.~$\mathcal{M}_d^{sa}(\mathbb C)$ shall be endowed with the $\|\cdot\|_M$, resp.~the operator norm (or the Schatten-$\infty$ norm, $\mathcal{S}_{\infty}$). Note that there is an abuse of notation here: we shall use $\|\cdot\|_M$ to denote norms on $\mathbb R^N$ and on $\mathbb R^N \otimes \mathcal{M}_d^{sa}(\mathbb C)$; the situation will be clear from the context. We shall now investigate the properties of the $\|\cdot\|_M$ norm with respect to this tensor product structure. We will consider that for given $N$-tuple of observables $(A_1, A_2, \ldots, A_N)$, we associate the tensor
	\begin{equation*}
		A := \sum_{x=1}^N e_x \otimes A_x \in \mathbb R^N \otimes \mathcal M_d^{sa}(\mathbb C).
	\end{equation*}
	
	\begin{definition}
		Given $p\in \mathbb{R}^N$, we define the following quantity:
		$$\|p\|_M:=\sum_{y=1}^N\bigg|\sum_{x=1}^N\,M_{xy}\,p_x\bigg| = \|M^\top p\|_1.$$
	\end{definition}
	In the lemma below we will show that $\|\cdot\|_M$ is a norm. 
	\begin{lemma}\label{M-norm}
		Given an invertible matrix $M$, the function $\mathbb{R}^N \ni p \mapsto \|p\|_M$ is a norm.
	\end{lemma}
	\begin{proof}
		Obviously we have $\|\alpha\, p\|_M=|\alpha|\,\|p\|_M$ for all $\alpha\in \mathbb{R}$.\\
		Now we will show that $\|p\|_M=0\implies p=0$\\
		$$\|p\|_M=\sum_{y=1}^N\bigg|\sum_{x=1}^NM_{xy}\,p_x\bigg|=0\implies\sum_{x=1}^NM_{xy}\, p_x=0\iff (M^{\top}p)_y=0$$\\
		by using the assumption that $M$ is invertible we have necessarily $p=0$, which ends the proof of $\|p\|_M\geq0$.\\
		Now we prove the triangle inequality $\|p+p'\|_M\leq\|p\|_M+\|p'\|_M$.\\
		Let's consider $$\|p+p'\|_M=\sum_{y=1}^N\bigg|\sum_{x=1}^NM_{xy}\,(p_x+p'_x)\bigg|
		\leq\sum_{y=1}^N\bigg|\sum_{x=1}^NM_{xy}\,p_x\bigg|+\sum_{y=1}^N\bigg|\sum_{x=1}^N\,M_{xy}\,p'_x\bigg|=\|p\|_M+\|p'\|_M$$
		Thus we have shown that $\|\cdot\|_M$ is a norm.
	\end{proof}
	
	By the Lemma \ref{M-norm}, we endow $\R^N$ with the norm $\|\cdot\|_M$, obtaining a Banach space $(\R^N,\|\cdot\|_M)$. In the following, we shall investigate the dual space of $(\R^N,\|\cdot\|_M)$. For that we shall compute the dual norm of $\|\cdot\|_M$ denoted by $\|\cdot\|^*_M$. 
	\begin{proposition}
		The dual norm $\|\cdot\|^*_M$ is given by: 
		$$\forall p \in \mathbb R^N, \qquad \|p\|_M^*=\max_y\Big|\sum_{z=1}^N\,(M^{-1})_{yz}\,p_z\Big| = \|M^{-1}p\|_\infty.$$
	\end{proposition}
	\begin{proof}
		Let $q,p\in\R^N$ we have 
		\begin{align*}
			|\langle p,q\rangle|=\Big|\sum_{x=1}^N\,p_x\,q_x\Big|&=\Big|\sum_{x,y,z=1}^N q_x\,M_{xy}\,M^{-1}_{yz}p_z\Big|=\Big|\sum_{y=1}^N\Big(\sum_{x=1}^N\,M_{xy}\,q_y\Big)\Big(\sum_{z=1}^N\,M^{-1}_{yz}p_z\Big)\Big|\\
			&\leq\sum_{y=1}^N\Big|\sum_{x=1}^N\,M_{xy}\,q_x\Big|\Big|\sum_{z=1}^N\,M^{-1}_{yz}\,p_z\Big|\leq \Big(\max_y\Big|\sum_{z=1}^N\,M^{-1}_{yz}\,p_z\Big|\Big)\,\sum_{y=1}^N\Big|\sum_{x=1}^N\,M_{xy}q_x\Big|\\
			&=\max_y\Big|\sum_{z=1}^N\,M^{-1}_{yz}\,p_z\Big|\,\|q\|_M.
		\end{align*}
		where we have used in the second equality that $M\cdot M^{-1}=I$. By taking the supremum over $\|q\|_M \leq 1$, we have shown that $\|p\|_M^* \leq \|M^{-1}p\|_\infty$. To show the converse inequality, note that
		$$\max_y\Big|\sum_{z=1}^N\,M^{-1}_{yz}\,p_z\Big| = \langle p, q\rangle, \qquad \text{for} \qquad q_z = \epsilon (M^{-1})_{y_0z}$$
		for some $y_0 \in [d]$ achieving the maximum, and $\epsilon = \pm 1$. In order to conclude, we have to establish that $\|q\|_M \leq 1$. Indeed, we have 
		$$\|q\|_M = \sum_y \Big| \sum_x M_{xy} \epsilon (M^{-1})_{y_0x}\Big| = \sum_y \big| (M^{-1}M)_{y_0y} \big| = 1.$$
	\end{proof}
	
	The Banach space $(\R^N,\|\cdot\|^*_M)$ is the dual of $(\R^N,\|\cdot\|_M)$: $(\R^N,\|\cdot\|_M)^*=(\R^N,\|\cdot\|^*_M)$. Now, we are ready to show the main theorem of this section, that the norm $\R^N\otimes\M_d^{sa}(\C)\ni A\to\|A\|_M$ is a tensor norm (or a reasonable crossnorm) in the sense of Definition \ref{def: reasonnable norm}. To this end, using Proposition \ref{Prop: reasonnable norm}, it suffices to show that:
	$$\|A\|_{\R^N\otimes_{\epsilon}\M^{sa}_d(\C)}\leq\|A\|_M\leq\|A\|_{\R^N\otimes_{\pi}\M^{sa}_d(\C)}$$
	where $\R^N$ is endowed with the norm $\|\cdot\|_M$ and $\M_d^{sa}$ with the norm $\|\cdot\|_{\infty}$.
	Before we show that $\|A\|_M$ is a tensor norm, we shall show the following proposition for tensors of rank one $A=p\otimes B\in\R^N\otimes\M_d^{sa}$.
	\begin{proposition}\label{prop:tensor-norm-M}
		Given $A\in\mathbb R^N \otimes \mathcal M_d^{sa}(\mathbb C)$ with $\mathbb{R}^N$ and $M_d^{sa}(\mathbb C)$ are endowed with $\|\cdot\|_{M}$ and the natural operator norm respectively. Given the particular decomposition $A=p\otimes B$ with $p\in(\mathbb{R}^N,\|\cdot\|_M)$ and $B\in(\mathcal M_d^{sa}(\mathbb C),\|\cdot\|_{\infty})$, one has $$\|p\otimes B\|_M=\|p\|_M\|B\|_{\infty}.$$
	\end{proposition}
	\begin{proof}
		Given $A=p\otimes B$ one has
		$$\|p\otimes B\|_M=\lambda_{\max}\bigg[\sum_{y=1}^N\bigg|\sum_{x=1}^NM_{xy}\,p_x B\bigg|\bigg]=\lambda_{\max}[|B|]\sum_{y=1}^N\bigg|\sum_{x=1}^NM_{xy}\,p_x\bigg|=\|B\|_{\infty}\|p\|_M.$$
		Above, we have used the following fact: for selfadjoint matrices $B$, 
		$$\|B\|_\infty = \max_{\lambda \text{ eig.~of $B$}} |\lambda| = \max_{\lambda \text{ eig.~of $|B|$}} \lambda = \lambda_{\max}[|B|].$$
	\end{proof}
	
	We now state and prove the following important result, establishing that the norm $\|\cdot\|_M$ is indeed a tensor norm. 
	
	\begin{theorem}\label{thm:M-tensor-norm}
		For a fixed $N$-input, 2-output invertible non-local game $M$, the quantity $\|\cdot\|_M$ introduced in Definition \ref{def:A-M}, which characterizes the largest quantum bias of the game $M$ when one fixes Alice's dichotomic measurements, is a reasonable crossnorm on $\mathcal M^{sa}_d(\mathbb C)^N \cong \mathbb R^N  \otimes \mathcal M_d^{sa}(\mathbb C)$: $$\|A\|_{\R^N\otimes_{\epsilon}\M^{sa}_d(\C)}\leq\|A\|_M\leq\|A\|_{\R^N\otimes_{\pi}\M^{sa}_d(\C)}$$
		with $(\R^N,\|\cdot\|_M)$ and $(\M_d^{sa}(\C),\|\cdot\|_{\infty})$.
	\end{theorem}
	Before we give the proof of the theorem we recall the definitions of the projective and the injective norms in our setting:
	\begin{align*}
		\|A\|_{\R^N\otimes_{\pi}\M^{sa}_d(\C)}&:=\inf\Big\{\sum_{i=1}^k\|p_i\|_M\,\|X_i\|_{\infty},\,A=\sum_{i=1}^k\,p_i\otimes X_i\Big\}.\\
		\|A\|_{\R^N\otimes_{\epsilon}\M^{sa}_d(\C)}&:=\sup\Big\{\langle\pi\otimes\alpha,A\rangle;\,\|\pi\|^*_M\leq 1,\,\|\alpha\|_1\leq1\Big\}.
	\end{align*}
	with $\M_d^{sa}(\C)\ni\alpha\to\|\alpha\|_1=\Tr|\alpha|$ is the Schatten 1-norm (or the nuclear norm). 
	\begin{proof}
		We shall prove first the easy direction: $\|A\|_M\leq\|A\|_{\R^N\otimes_{\pi}\M^{sa}_d(\C)}$. Let us consider a decomposition $A=\sum_{i=1}^k\,p_i\otimes X_i$. We have 
		$$\|A\|_M=\Big\|\sum_{i=1}^k\,p_i\otimes X_i\Big\|_M\leq\sum_{i=1}^k\|p_i\otimes X_i\|_M=\sum_{i=1}^k\|p_i\|_M\|X_i\|_{\infty},$$
		where the factorization property follows by Proposition \ref{prop:tensor-norm-M}.
		Hence we have $$\|A\|_M\leq\|A\|_{\R^N\otimes_{\pi}\M^{sa}_d(\C)}.$$
		We shall now prove that $\|A\|_{\R^N\otimes_{\epsilon}\M^{sa}_d(\C)}\leq\|A\|_M$. Let $\alpha=\pm\ketbra{\phi}{\phi} \in(\M_d^{sa}$ be an extremal point of the unit ball of the $\mathcal S_1$ space and $\pi\in(\R^N,\|\cdot\|^*_M)$. We have 
		\begin{align*}
			|\langle \pi\otimes\alpha,A\rangle|&=\Big|\Big\langle\alpha,\sum_{x=1}^N\,\pi_x\,A_x\Big\rangle\Big|=\Big|\Big\langle\alpha,\sum_{x,y,z=1}^N\, A_z\,M_{zy}\,M^{-1}_{yx}\,\pi_x\Big\rangle\Big|\\
			&=\sum_{y=1}^N\Big|\sum_{z=1}^N\,M_{zy}\,\langle \alpha,A_z\rangle\Big|\Big|\sum_{x=1}M^{-1}_{yx}\pi_{x}\Big|\leq\sum_{y=1}^N\Big|\sum_{z=1}^N\,M_{zy}\,\langle \alpha,A_z\rangle\Big|\max_y\Big|\sum_{x=1}M^{-1}_{yx}\pi_{x}\Big|\\
			&=\|\pi\|^*_M\,\sum_{y=1}^N\Big|\sum_{z=1}^N\,M_{zy}\,\langle \alpha,A_z\rangle\Big|=\|\pi\|^*_M\,\sum_{y=1}^N\Big|\sum_{z=1}^N\,\langle \alpha,M_{zy}\,A_z\rangle\Big|\\
			&=\|\pi\|^*_M\,\sum_{y=1}^N\Big|\sum_{z=1}^N\,\Tr\Big[ \alpha\,M_{zy}\,A_z\Big]\Big|=\|\pi\|^*_M\,\sum_{y=1}^N\Big|\sum_{z=1}^N\,\bra{\phi} \,M_{zy}\,A_z\ket{\phi}\Big|\\
			&=\|\pi\|^*_M\,\sum_{y=1}^N\Big|\bra{\phi} \,\Big(\sum_{z=1}^N\,M_{zy}\,A_z\Big)^+\ket{\phi}-\bra{\phi} \,\Big(\sum_{z=1}^N\,M_{zy}\,A_z\Big)^-\ket{\phi}\Big|\\
			&\leq\|\pi\|^*_M\,\sum_{y=1}^N\Big[\Big|\bra{\phi} \,\Big(\sum_{z=1}^N\,M_{zy}\,A_z\Big)^+\ket{\phi}\Big|+\Big|\bra{\phi} \,\Big(\sum_{z=1}^N\,M_{zy}\,A_z\Big)^-\ket{\phi}\Big|\Big]\\
			&=\|\pi\|^*_M\sum_{y=1}^N\bra{\phi}\Big|\sum_{z=1}^N\,M_{zy}\,A_z\Big|\ket{\phi}.
		\end{align*}
		By taking the supremum $\|\pi\|^*_M\leq1 $ and $\|\alpha\|_{\mathcal{S}_1}\leq1$ on the last expression we have:
		$$\sup\{|\langle\pi\otimes\alpha,A\rangle|;\,\|\pi\|^*_M\leq1\,,\,\|\alpha\|_{\mathcal{S}_1}\leq1\}\leq\sup_{\|\phi\|=1}\sum_{y=1}^N\bra{\phi}\Big|\sum_{z=1}^N\,M_{zy}\,A_z\Big|\ket{\phi}=\lambda_{\max}\Big[\sum_{y=1}^N\Big|\sum_{z=1}^N\,M_{zy}\,A_z\Big|\Big].$$
		Hence we have 
		$$\|A\|_{\R^N\otimes_{\epsilon}\M^{sa}_d(\C)}\leq\|A\|_M$$
		which ends the proof of the theorem.
	\end{proof}

	\section{Dichotomic measurement compatibility via tensor norms}\label{sec: compatibility-norm}
	
	Having addressed in the previous sections the maximum value of a non-local game $M$ with fixed dichotomic observables on Alice's side $A$, we now turn to the second object of our study, quantum measurement (in-)compatibility. We characterize compatibility of dichotomic measurements (or quantum ) using tensor norms, following \cite{bluhm2022incompatibility}. Recall that we associate a dichotomic POVM $(E, I-E)$ to the corresponding observable $A = E - (I-E) = 2E - I$. In other words, the effect $E$ corresponds to the ``$+1$'' outcome, while the effect $I-E$ corresponds to the other outcome, ``$-1$''. This way, the set of dichotomic POVMs is mapped to the set of selfadjoint operators $-I \leq A \leq I$. 
	
	To a $N$-tuple of observables $(A_1, A_2, \ldots, A_N)$, we associate the tensor
	\begin{equation}
		A := \sum_{i=1}^N e_i \otimes A_i \in \mathbb R^N \otimes \mathcal M_d^{sa}(\mathbb C).
	\end{equation}
	
	The (real) vector spaces $\mathbb R^N$, resp.~$\mathcal M_d^{sa}(\mathbb C)$ shall be endowed with the $\ell_\infty$, resp.~the operator norm (or the Schatten-$\infty$ norm, $S_{\infty}$). On the tensor product space 
	$$\mathbb R^N \otimes \mathcal M_d^{sa}(\mathbb C) \cong \left[ \mathcal M_d^{sa}(\mathbb C) \right]^N$$
	we shall consider two tensor norms: the \emph{injective norm}
	\begin{equation}\label{eq:epsilon-norm}
		\|X = (X_1, X_2, \ldots, X_N)\|_\epsilon = \max_{i=1}^N \|X_i\|_\infty
	\end{equation}
	and the compatibility norm, which was introduced in \cite[Proposition 9.4]{bluhm2022incompatibility}. We review next its definition and its basic properties, in order to make the presentation self-contained. We note however that the situation considered in \cite[Section 9]{bluhm2022incompatibility} is more general, going beyond the case of quantum mechanics. 
	
	\begin{definition}\label{def:compatibility-norm}
		For a tensor $X \in \mathbb R^N \otimes \mathcal M_d^{sa}(\mathbb C)$, we define the following quantity, which we call the \emph{compatibility norm} of $X$: 
		\begin{equation}\label{eq:def-norm-c}
			\|X\|_c := \inf \left\{ \Big\|\sum_{j=1}^K H_j\Big\|_\infty \, : \, X = \sum_{j=1}^K z_j \otimes H_j, \, \text{ s.t. } \, \forall j \in [K], \, \|z_j\|_\infty \leq 1 \text{ and } H_j \geq 0\right\}.
		\end{equation}
	\end{definition}
	
	Note that in the case of a single matrix ($N=1$) we have $\|(X_1)\|_c = \|X_1\|_\infty$, and that, in general, we have
	$$\|X\|_c = \inf \left\{ t \, : \, X = \sum_{j=1}^K z_j \otimes H_j, \, \text{ s.t. } \, \sum_{j=1}^K H_j = t I_d \text{ and }  \forall j \in [K], \, \|z_j\|_\infty = 1, H_j \geq 0\right\}.$$
	Indeed, the condition $\|z_j\|=1$ can be imposed by replacing a non-zero term $z_j \otimes H_j$ by $z_j/\|z_j\|_\infty \otimes \|z_j\|_\infty H_j$, while the condition $\sum_j H_j = tI_d$ can be imposed by adding the term $0 \otimes (t I_d - \sum_j H_j)$ to the decomposition.
	
	\begin{proposition}\label{prop:c-norm-tensor}
		The $\|\cdot\|_c$ quantity is a tensor norm on $(\mathbb R^N, \|\cdot\|_\infty) \otimes (\mathcal M_d^{sa}(\mathbb C), \|\cdot\|_\infty)$.
	\end{proposition}
	\begin{proof}
		Let us start with the triangle inequality, $\|A+B\|_c \leq \|A\|_c + \|B\|_c$. Consider optimal decompositions
		\begin{align*}
			A &= \sum_j z_j \otimes H_j\\
			B &= \sum_k w_k \otimes T_k
		\end{align*}
		such that 
		$$\|A\|_c = \Big\|\sum_j H_j \Big\| \qquad \text{and} \qquad \|B\|_c = \Big\|\sum_k T_k \Big\|.$$ 
		Then, 
		$$A+B = \sum_j z_j \otimes H_j + \sum_k w_k \otimes T_k$$
		is a valid decomposition for $A+B$, hence
		$$\|A+B\|_c \leq \Big\|\sum_j H_j + \sum_k T_k\Big\| \leq \Big\|\sum_j H_j \Big\| + \Big\| \sum_k T_k\Big\| = \|A\|_c + \|B\|_c.$$
		
		The scaling equality $\|\lambda A \|_c = |\lambda| \|A\|_c$ is straightforward, and left to the reader. Consider now $A$ such that $\|A\|_c = 0$. Then, for all $\epsilon > 0$, there is a finite decomposition $A = \sum_j z_j \otimes H_j$ such that $\| \sum_j H_j\| \leq \epsilon$. We have then, for all $x \in [N]$, 
		$$\|A_x\| = \Big\|\sum_j z_j(x) H_j\Big\| \leq \Big\|\sum_j |z_j(k)| H_j\Big\| \leq \Big\|\sum_j H_j\Big\| \leq \epsilon.$$
		Taking $\epsilon \to 0$ shows that $A_x = 0$ for all $x$, and thus $A = 0$.
		
		The fact that the compatibility norm is bounded by the injective and projective norms is established in \cite[Proposition 3.3]{bluhm2022tensor}. Finally, let us show that $\|\cdot\|_c$ factorizes on simple tensors. To this end, consider a (non-zero) product tensor $A = w \otimes T$ with $\|w\|_\infty = 1$ (this can always be enforced by absorbing the norm of $w$ into $T$). On the one hand, we have 
		$$\|A\|_c \leq \|T\| = \|w\|_\infty \|T\|,$$
		establishing one inequality. Consider now an optimal decomposition 
		$$w \otimes T = \sum_j z_j \otimes H_j$$
		with $\|z_j\|_\infty \leq 1$, $H_j \geq 0$, and $\|w \otimes T\|_c = \|\sum_j H_j\|$. Consider an index $k \in [N]$ such that $\|w\|_\infty =  |w(k)|$. We have then $w(k) T = \sum_j z_j(k) H_j$ and thus 
		$$\|w\|_\infty \|T\| = \Big\| \sum_j z_j(k) H_j \Big \| \leq \Big\| \sum_j |z_j(k)| H_j \Big \| \leq \Big\| \sum_j  H_j \Big \| = \|w \otimes T\|_c,$$
		finishing the proof.
	\end{proof}
	
	We specialize now \cite[Theorem 9.2]{bluhm2022incompatibility} to the case of quantum mechanics, showing that the compatibility norm from Definition \ref{def:compatibility-norm}.
	
	\begin{theorem}\label{thm:c-norm-compatibility}
		Let $A = (A_1, \ldots, A_N)$ be a $N$-tuple of self-adjoint $d \times d$ complex matrices. Then:
		\begin{enumerate}
			\item $A$ is a collection of dichotomic quantum observables (i.e.~$\|A_i\|_\infty \leq 1$ $\forall i$) if and only if $\|A\|_\epsilon \leq 1$, where $\|\cdot\|_\epsilon$ quantity is the $\ell_\infty^N \otimes_\epsilon S_\infty^d$ tensor norm.
			\item $A$ is a collection of \emph{compatible} dichotomic quantum observables if and only if $\|A\|_c \leq 1$.
		\end{enumerate}
	\end{theorem}
	\begin{proof}
		The first statement is a direct consequence of ~\eqref{eq:epsilon-norm}. For the second statement, we shall prove the two implications separately. 
		
		First, consider compatible dichotomic observables $A_1, \ldots, A_N$, and their joint POVM $X$, having  $X_\epsilon \geq 0$ indexed by sign vectors $\epsilon \in \{\pm 1\}^N$, such that
		$$\forall i \in [N], \, \forall s \in \{\pm 1\}, \qquad E^s_i = \frac{I_d + s A_i}{2} = \sum_{\epsilon \in \{\pm 1\}^N \, : \, \epsilon_i = s} X_\epsilon.$$
		In particular, we have, for all $i \in [N]$, 
		$$A_i = -I_d + 2 \sum_{\epsilon \in \{\pm 1\}^N \, : \, \epsilon_i = +1} X_\epsilon$$
		and thus
		\begin{align*}
			A &= \sum_{i=1}^N e_i \otimes A_i = \sum_{i=1}^N (-e_i) \otimes I_d + 2 \sum_{\epsilon \in \{\pm 1\}^N} \left(\sum_{i \, : \, \epsilon_i = +1} e_i \right) \otimes X_\epsilon \\
			&= \sum_{\epsilon \in \{\pm 1\}^N} \left( 2\sum_{i \, : \, \epsilon_i = +1} e_i - \sum_i e_i \right) \otimes X_\epsilon \\
			&= \sum_{\epsilon \in \{\pm 1\}^N} \Big(\underbrace{ \sum_i \epsilon_i e_i }_{=:z_\epsilon}\Big) \otimes X_\epsilon.
		\end{align*}
		We have thus obtained above a decomposition of the tensor $A$ with $2^N$ terms, $\|z_\epsilon\|_\infty=1$ and $\sum_\epsilon X_\epsilon = I_d$, proving that $\|A\|_c \leq 1$.
		
		For the reverse implication, start with a decomposition $A = \sum_j z_j \otimes H_j$ with $\|z_j\|_\infty \leq 1$, $H_j \geq 0$ and $\sum_j H_j = I_d$. One can recover the observables and the  from this decomposition: 
		$$A_i = \sum_j z_j(i) H_j \quad \text{ and } \quad E_i^\pm = \sum_j \frac{1 \pm z_j(i)}{2} H_j.$$
		One recognizes in the expression above the description of the compatibility of the POVMs $(E^+_i, E^-_i)_{i \in [N]}$ as post-processing from Proposition \ref{prop:compatibility-postprocessing}:
		$$E^\pm_i = \sum_j p_i(\pm | j) H_j,$$
		where the conditional probabilities $p_i$ are given by
		$$p_i(\pm | j) = \frac{1 \pm z_j(i)}{2} \in [0,1].$$
	\end{proof}
	
	The compatibility norm of a tensor $A$ is related to the noise parameter $\Gamma$ from Definition \ref{def:Gamma}. The following proposition provides an \emph{operational interpretation} of the compatibility norm $\|A\|_c$, as the inverse of the minimal quantity of white noise that needs to be mixed in the measurements $A$ in order to render them compatible. 
	
	\begin{proposition}\label{prop: compatibility norm and the noise threshold}
		For any $N$-tuple of observables $A = (A_1, A_2, \ldots, A_N) \neq 0$, 
		$$\Gamma(A) = \frac{1}{\|A\|_c}.$$
	\end{proposition}
	\begin{proof}
		Note first that, on the level of observables, adding noise to a dichotomic measurement corresponds to scaling: 
		$$A^\eta = \left[\eta E + (1-\eta)\frac I 2\right] - \left[I - \eta E - (1-\eta)\frac I 2\right] = 2\eta E - \eta I = \eta A.$$
		Hence, 
		\begin{align*}
			\Gamma(A) &= \max\{\eta \, : \, (A_1^\eta, \ldots, A_N^\eta) \text{ compatible}\}\\
			&= \max\{\eta \, : \, \|A_1^\eta, \ldots, A_N^\eta\|_c \leq 1\}\\
			&= \max\{\eta \, : \, \eta \|A_1, \ldots, A_N\|_c \leq 1\}\\
			&= \frac{1}{\|A\|_c}.
		\end{align*}
	\end{proof}
	
	\begin{example}
		Let us consider the example of the unbiased Pauli measurements, 
		$$\frac 1 2 (I_2 \pm x \sigma_X), \quad \frac 1 2 (I_2 \pm y \sigma_Y), \quad \frac 1 2 (I_2 \pm z \sigma_Z),$$
		where $(x,y,z) \in [0,1]^3$ are real parameters describing the noise in the measurements. These three POVMs correspond to the observables 
		$$A_X = x \sigma_X, \quad A_Y = y \sigma_Y, \quad A_Z = z \sigma_Z.$$
		It is known \cite{busch86,brougham2007estimating} that these observables are compatible if and only if $x^2 +y^2+z^2 \leq 1$, hence 
		$$\|(A_X, A_Y, A_Z)\|_c = \sqrt{x^2+y^2+z^2} =  \|(x,y,z)\|_2.$$
	\end{example}
	
	\section{The relation between non-locality and incompatibility}\label{sec: non-locality and incompatibility}
	
	Having introduced the main conceptual definitions of the \emph{(in)compatibility norm} and \emph{M-Bell-locality norm} in the previous sections that formalize the compatibility and the non-locality physical notions for fixed measurements on Alice's side apparatus and invertible non-local games $M$, we bring together and compare the two norms. In this section we introduce \emph{the main theorems} of the paper. It was shown in \cite{wolf2009measurements} that the two notions are equivalent in the case of the CHSH game. Using the framework of tensor norms, we shall give a quantitative and precise answer to the following question:
	
	\medskip
	
	\emph{When is measurement incompatibility equivalent to non-locality for general games?} 
	
	\medskip
	
	\noindent It turns out that the answer to this question is given by a comparison between the \emph{compatibility norm} and the \emph{M-Bell-locality norm}. For the reader's convenience, we recall the definitions of the two tensor norms that we introduced in the Sections \ref{sec:non-locality-norm} and \ref{sec: compatibility-norm}, in relation to, respectively, Bell inequality violations and measurement incompatibility.
	\begin{itemize}
		\item The \emph{$M$-Bell-locality norm} (see Definition \ref{def:A-M} and Theorem \ref{thm:M-tensor-norm})
		$$\|A\|_{M}:=\sup_{\|\psi\| = 1} \sup_{\|B_y\| \leq 1} \Big \langle \psi \Big | \sum_{x,y = 1}^N\, M_{xy} \,A_x \otimes B_y \Big | \psi \Big \rangle=\lambda_{\max}\bigg[\sum_{y=1}^N \bigg| \sum_{x=1}^N  \, M_{xy}\,A_x\bigg|\bigg].$$
		
		\item The \emph{compatbility norm} (see Definition \ref{def:compatibility-norm} and Theorem \ref{thm:c-norm-compatibility})
		$$\|A\|_c := \inf \left\{ \Big\|\sum_{j=1}^K H_j\Big\|_\infty \, : \, A = \sum_{j=1}^K z_j \otimes H_j, \, \text{ s.t. } \, \forall j \in [K], \, \|z_j\|_\infty \leq 1 \text{ and } H_j \geq 0\right\}.$$
	\end{itemize}
	
	In what follows, we shall compare these two norms, in order to relate, in a quantitative manner, the two fundamental physical phenomena of Bell non-locality and measurement incompatibility. 
	
	We start with a reformulation, using the language of tensor norms, of the following well established fact: an observed \emph{violation of the Bell inequality} $M$ implies necessarily the \emph{incompatibility} of Alice's measurements. Mathematically, this corresponds to upper bounding the $M$-Bell-locality norm of Alice's measurements by their compatibility norm. 
	
	\begin{theorem}\label{thm:M-leq-rho}
		Consider a $N$-input, 2-output non-local inevrtible game $M$, corresponding to a matrix $M \in \mathcal M_N(\mathbb R)$. Then, for any $N$-tuple of self-adjoint matrices $A = (A_1, \ldots, A_N)$, we have
		\begin{equation}\label{eq:M-leq-rho}
			\|A\|_M\leq\|A\|_{c}\|M\|_{\ell_1^N \otimes_{\epsilon} \ell_1^N} = \|A\|_{c} \, \beta(M).
		\end{equation}
		In particular, if Alice's measurements $A$ are $M$-Bell-non-local (in the sense of Definition \ref{def:Bell-local}), then they must be incompatible. 
	\end{theorem}
	\begin{proof}
		Let us consider the optimal decomposition $\|A\|_{c}=\|\sum_{j=1}^N C_j\|_{\infty}$ with $A=\sum_{j=1}^N z_j\otimes C_j$, $\|z_j\|_{\infty}\leq 1$ and $C_j\geq0$ for all $j$. Thus we have $A_x=\sum_{j=1}^Nz_j(x)C_j$.\\
		We compute the upper bound of the $M$-Bell-locality norm $\|A\|_M$.
		\begin{align*}
			\|A\|_M&=\lambda_{\max}\left[\sum_{y=1}^N\bigg|\sum_{x=1}^N    
			\,M_{xy}\,A_x\bigg|\right]=\lambda_{\max}\left[\sum_{y=1}^N\bigg|\sum_{j=1}^N\sum_{x=1}^N \,M_{xy}\,z_j(x)\,C_j\bigg|\right]\\
			&\leq \lambda_{\max}\left[\sum_{j=1}^N\sum_{y=1}^N\bigg|\sum_{x=1}^N\, M_{xy}\,z_j(x)\bigg|\,C_j\right]=\lambda_{\max}\left[\sum_{j=1}^N\sum_{y=1}^N\sum_{x=1}^N \,\epsilon_y \,M_{xy}\,z_j(x)\,C_j\right]
		\end{align*}
		We have used  $$ \sum_{y=1}^N\bigg|\sum_{x=1}^N\, M_{xy}\,z_j(x)\bigg|=\sum_{y=1}^N\sum_{x=1}^N \,\epsilon_y \,M_{xy}\,z_j(x)$$
		with $\epsilon=\{\pm 1\}^N$.
		
		Then we have 
		\begin{align*}
			\|A\|_M&\leq\lambda_{\max}\left[\sum_{j=1}^N\sum_{y=1}^N\sum_{x=1}^N\,                   \epsilon_y\,M_{xy}\,z_j(x)\,C_j\right]\leq\lambda_{\max}\left[\sum_{j=1}^N \,C_j\,\|M\|_{\ell_1 \otimes_\epsilon \ell_1}\right]\\
			&=\lambda_{\max}\left[\sum_{j=1}^N C_j\right] \|M\|_{\ell_1 \otimes_\epsilon \ell_1}=\|A\|_{c} \|M\|_{\ell_1 \otimes_\epsilon \ell_1}
		\end{align*}
		where $\|M\|_{\ell_1 \otimes_\epsilon \ell_1}=\underset{\|\epsilon\|_{\infty}\leq 1,\|z_j\|_{\infty}\leq 1}{\sup} \langle M,z_j\otimes \epsilon \rangle$.   
	\end{proof}
	
	In the following we will show, for \emph{invertible Bell functionals}, that the compatibility is upper bounded by the $M$-Bell-locality norm. 
	
	\begin{theorem}\label{thm:rho-leq-M}
		Consider a $N$-input, 2-output non-local game $M$, corresponding to an \emph{invertible} matrix $M \in \mathcal M_N(\mathbb R)$. Then, for any $N$-tuple of self-adjoint matrices $A = (A_1, \ldots, A_N)$, we have
		\begin{equation}\label{eq:rho-leq-M}
			\|A\|_c \leq \|A\|_M \|M^{-1}\|_{\ell^N_\infty \otimes_\epsilon \ell^N_\infty}.
		\end{equation}
	\end{theorem}

	\begin{proof}
		Let us consider $$C_y=\sum_{x=1}^N M_{xy} A_x=(M^\top A)_y\implies A_x=((M^\top)^{-1}C)_x=\sum_{y=1}^N(M^{-1})_{y,x}C_y.$$
		Let us also consider the following decomposition of $A=\sum_{x=1}^Ne_x\otimes A_x$ with $e_x$ the canonical basis vectors.
		We have
		$$A=\sum_{x=1}^Ne_x\otimes A_x=\sum_{y=1}^N\sum_{x=1}^N(M^{-1})_{y,x}e_x\otimes C_y$$
		$$=\sum_{y=1}^N\left[\sum_{x=1}^N(M^{-1})_{y,x}e_x\right]\otimes C^+_y+\sum_{y=1}^N\left[-\sum_{x=1}^N(M^{-1})_{y,x}e_x\right]\otimes C^-_y$$
		$$=\sum_{y=1}^N e'_y\otimes C^+_y+\sum_{y=1}^N-e'_y\otimes C^-_y$$
		where we have decomposed $C_y=C^+_y-C^-_y$ into positive and negative parts $C_y^\pm \geq 0$ for all $y \in [N]$ and $e'_y:=\sum_{x=1}^N(M^{-1})_{y,x}e_x$. 
		
		Observe that 
		$$\|e'_y\|_{\infty}=\Big\|\sum_{x=1}^N(M^{-1})_{y,x}e_x\Big\|_{\infty}=\Big\|(M^{-1}e)_y\Big\|_{\infty}\leq\|M^{-1}\|_{\infty}\|e\|_{\infty}=\|M^{-1}\|_{\ell^N_\infty \otimes_\epsilon \ell^N_\infty}$$
		where we recall by the Definition \ref{def: injective norm} that  $$\|M^{-1}\|_{\ell^N_\infty \otimes_\epsilon \ell^N_\infty}:=\sup_{a,b\in\mathbb B(\ell_1^N(\R))}|\langle M^{-1},a\otimes b\rangle|=\max_{i,j}|(M^{-1})_{i,j}|=\|M^{-1}\|_{\infty}.$$
		With the norm formulation above, $\|M^{-1}\|_{\ell^N_\infty \otimes_\epsilon \ell^N_\infty}$ and $\|M^{-1}\|_{\infty}$ are equal if we consider $M^{-1}$ in its matrix or tensor representations. Hence, we have
		$$\|e'_y\|_{\infty}\leq\|M^{-1}\|_{\ell^N_\infty \otimes_\epsilon \ell^N_\infty}.$$
		We consider now the normalised vectors
		$$a_y:=\frac{e_y}{\|M^{-1}\|_{\ell^N_\infty \otimes_\epsilon \ell^N_\infty}}\in\mathbb{B}(\ell_{\infty}(\mathbb{R}^N))$$
		
		By normalising the vectors one has 
		$$A=\sum_{y=1}^N\, \|M^{-1}\|_{\ell^N_\infty \otimes_\epsilon \ell^N_\infty}\, a_y\otimes C^+_y\,-\,\sum_{y=1}^N\,\|M^{-1}\|_{\ell^N_\infty \otimes_\epsilon \ell^N_\infty}\, a_y\otimes C^-_y$$
		We recognize above a valid decomposition of the tensor $A$ as in Eq.~\eqref{eq:def-norm-c}.
		Hence
		\begin{align*}
			\|A\|_c&\leq\Big\|\sum_{y=1}^N\,\|M^{-1}\|_{\ell^N_\infty\otimes_\epsilon             \ell^N_\infty}\,(C^+_y\,+\,C^-_y)\Big\|_{\infty}\\
			&=\|M^{-1}\|_{\ell^N_\infty \otimes_\epsilon \ell^N_\infty}\lambda_{max}\Big(\sum_{y=1}^N\Big|C_y\Big|\Big)=\|M^{-1}\|_{\ell^N_\infty \otimes_\epsilon \ell^N_\infty}\|A\|_M.
		\end{align*}
	\end{proof}
	
	Putting together Theorems \ref{thm:M-leq-rho} and \ref{thm:rho-leq-M}, we recover the main result from \cite{wolf2009measurements}: for $N=2$ and the CHSH matrix 
	$$M_{\textrm{CHSH}} = \frac 1 2 \begin{bmatrix}
		1 & 1 \\ 
		1 & -1 \end{bmatrix},$$
	we have 
	$$\beta(M_{\textrm{CHSH}}) = 1 \qquad  \text{ and } \qquad (M_{\textrm{CHSH}})^{-1} = \begin{bmatrix}
		1 & 1 \\ 
		1 & -1 \end{bmatrix}.$$
	It follows thus, from Eqs.~\eqref{eq:M-leq-rho} and \eqref{eq:rho-leq-M} that
	\begin{equation}\label{eq:c-CHSH-equal}
		\|\cdot\|_c = \|\cdot\|_{M_{\textrm{CHSH}}}.    
	\end{equation}
	
	\begin{remark}
		We have seen in Section \ref{sec:non-locality} that for the CHSH game we have $$\|A\|_{M_{\text{CHSH}}}=\frac{1}{\Gamma(A)}.$$
		One can see also that within Proposition \ref{prop: compatibility norm and the noise threshold} we have, with respect to the compatibility norm, 
		$$\|A\|_c=\frac{1}{\Gamma(A)}.$$
	\end{remark}

	In what follows, we compare the compatibility norm and the Bell-locality norm in two special settings, a modified CHSH game, and the pure correlation $I_{3322}$ game. As we shall see, in these situations, the two norms are different. 
	
	\bigskip

	Let us first consider the following modified CHSH game defined by the following matrix
	$$M_t :=\begin{bmatrix}
		1 & 1 \\ 
		1 & -t \end{bmatrix},$$
	where $t$ is a real parameter taking values in $\mathbb{R}\setminus\{-1\}$, such that the matrix $M_t$ is invertible. 
	We start by normalizing the matrix $M_t$ such that its classical bias $\beta$ is equal to 1. A simple calculation shows that
	$$\beta(M_t):=\begin{cases}3-t\qquad \text{for} \quad t\leq1\\
		1+t \qquad \text{for} \quad t>1
	\end{cases} = 2 + |t-1|.$$
	We are thus going to work with the normalized version
	$$M'_t =\frac{1}{\beta(M_t)}\begin{bmatrix}
		1 & 1 \\ 
		1 & -t \end{bmatrix},$$
	for which $\beta(M'_t) = 1$. 
	We consider the following pair of spin observables 
	$$A:=(\sigma_X,y\sigma_Y),$$
	where $y \in [-1,1]$ is a parameter we shall vary. These two observables correspond to, respectively, a sharp measurement in the eigenbasis of $\sigma_X$ and a noisy measurement in the eigenbasis of $\sigma_Y$. 
	
	In the following, we calculate $\|A\|_{c}$ and $\|A\|_{M'_t}$, for different values of the parameters $t$ and $y$. Since the $t=1$ value corresponds to the $CHSH$ game (for which $\|\cdot\|_c = \|\cdot\|_{M_{\textrm{CHSH}}}$), the compatibility norm reads
	$$\|A\|_{c}=\|A\|_{M'_{t=1}}=\|A\|_{M_{\textrm{CHSH}}}=\lambda_{\max}\bigg[\sum_{y=1}^2\bigg|\sum_{x=1}^2M_{xy}A_x\bigg|\bigg]=\frac{1}{2}\lambda_{\max}\big[|\sigma_X+y\sigma_Y|+|\sigma_X-y\sigma_Y|\big].$$
	A simple calculation shows that$$\|A\|_{c}=\sqrt{1+y^2}=:r.$$
	
	We now compute $\|A\|_{M'_t}$ for the normalized modified CHSH game:
	$$\|A\|_{M'_t}=\frac{1}{2+|t-1|}\lambda_{\max}\big[|\sigma_X+y\sigma_Y|+|\sigma_X-ty\sigma_Y|\big]=\frac{r_t+r}{2+|t-1|}.$$
	with $r_t:=\sqrt{1+(yt)^2}$; above, we have used the following fact: 
	$$\forall x,y \in \mathbb R, \qquad |x \sigma_X + y \sigma_Y| = \left| \begin{bmatrix} 0 & x - \mathrm{i} y \\ x + \mathrm{i} y & 0 \end{bmatrix} \right|  =\sqrt{x^2+y^2} \, I_2.$$
	. We plot the norm $\|A\|_{M'_t}$ in Figure \ref{fig:norme A_Mt}, the region $t$ and $y$ where Alice observes a Bell inequality violation $\|A\|_{M'_t}> \beta(M
	_t) = 1$ in Figure \ref{fig: Alice observe violation}, and the ratio of the two norms in Figure \ref{fig:modified-CHSH}. Note that the plot for $t=1$ corresponds to the CHSH game: the two norms are equal (see Eq.~\eqref{eq:c-CHSH-equal}). At $y=1$, Alice's measurements are sharp: $A = (\sigma_X, \sigma_Y)$. One observes violations of the game $M'_t$ for the parameter values
	$$\|A_{y=1}\|_{M'_t} > 1 \iff t > \frac{9-4\sqrt 2}{7}=:t_*.$$
	The values at $y=0$ also have a special meaning, since, in this case,
	$$A = (\sigma_X, 0) = (1,0) \otimes \sigma_X.$$
	By the tensor norm property of the compatibility norm (see Proposition \ref{prop:c-norm-tensor}), we have 
	$$\|A\|_c = \|(1,0)\|_{\ell_\infty} \cdot \|\sigma_X\|_{S_\infty} = 1.$$
	Similarly, the tensor norm property of the Bell-locality norm yields
	$$\|A\|_{M_t'} = \|(1,0)\|_{M_t'} \cdot \|\sigma_X\|_{S_\infty} = \frac{2}{\beta(M_t)} \cdot 1 = \frac{2}{2+|t-1|} \leq 1.$$
	In Figure \ref{fig:modified-CHSH}, the dashed curve corresponds to the limit $|t| \to \infty$, in which case
	$$\lim_{|t| \to \infty} \frac{\|A\|_{M_t'}}{\|A\|_c} = \frac{|y|}{\sqrt{1+y^2}}.$$
	Finally, the dotted line corresponds to the game $M_t'$ for $t=-1$. This game is not invertible, so the quantity $\|\cdot\|_{M_{-1}'}$ is not a norm. 
	\begin{figure}[htb!]
		\centering
		\includegraphics[width=\textwidth]{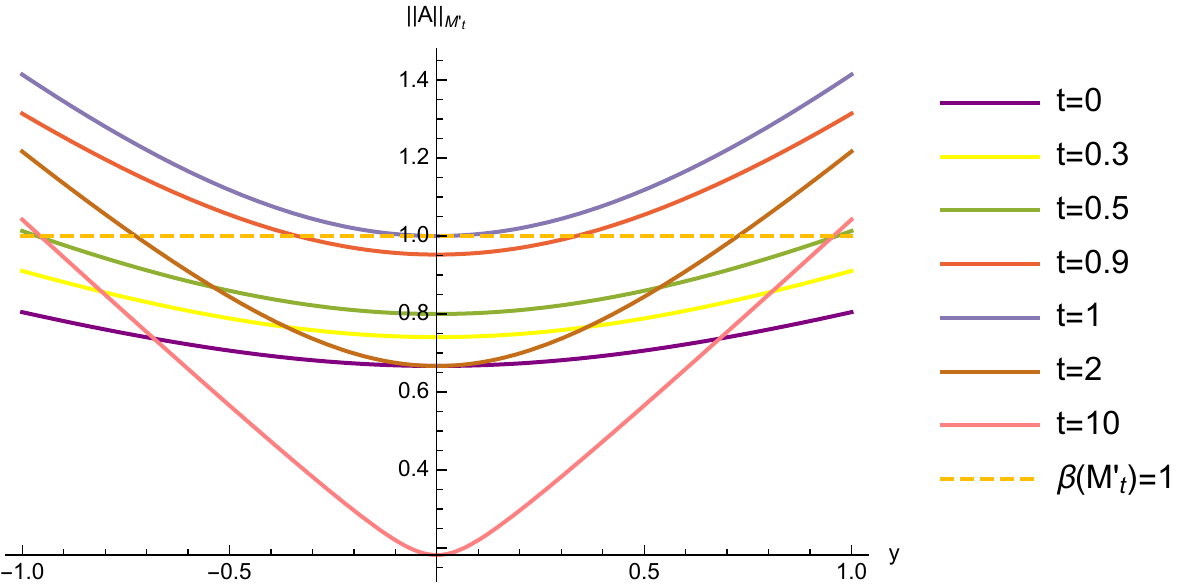}\qquad\qquad\qquad\qquad 
		\caption{The norm $\|A\|_{M'_t}$ for $y\in[-1,1]$ and different value of $t$. The measurements $A$ are $\sigma_X$ and $y\sigma_Y$, a noisy version of $\sigma_Y$. For $t=1$ (the CHSH game), one observes violations (i.e.~$\|A\|_{M'_t} > \beta(M'_t)=1$) for every value of $y \in [-1, 1]$.}
		\label{fig:norme A_Mt}
	\end{figure}
	
	\begin{figure}[htb!]
		\centering
		\includegraphics[width=0.7\textwidth]{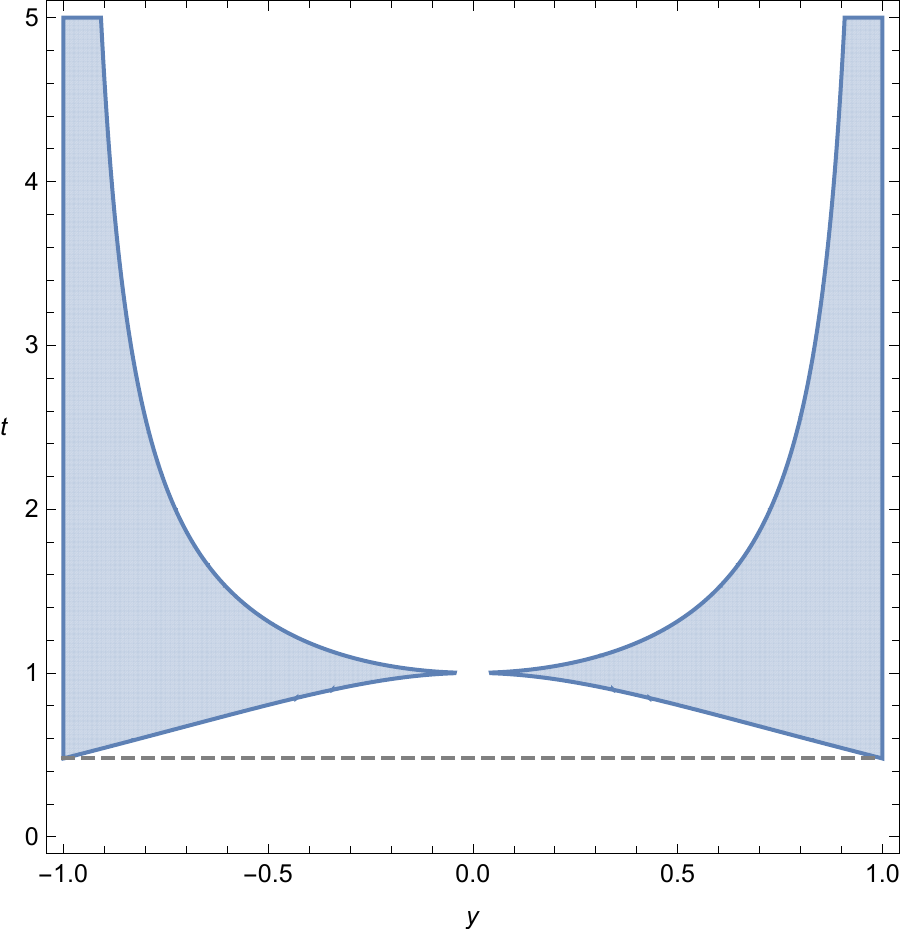}\qquad\qquad\qquad\qquad 
		\caption{The filled region corresponds to the set of parameters $(y,t)$ for which  Alice's measurements ($\sigma_X$ and $y\sigma_Y$) are Bell non-local (for the game $M'_t$):  $\|A\|_{M'_t}>1 = \beta(M'_t)$. Note that for $t > t_* = (9-4\sqrt 2)/7$, the game $M'_t$ does not allow quantum violations when Alice's measurements are of the form $A= (\sigma_X, y\sigma_Y).$}
		\label{fig: Alice observe violation}
	\end{figure}
	
	\begin{figure}[htb!]
		\centering
		\includegraphics[width=\textwidth]{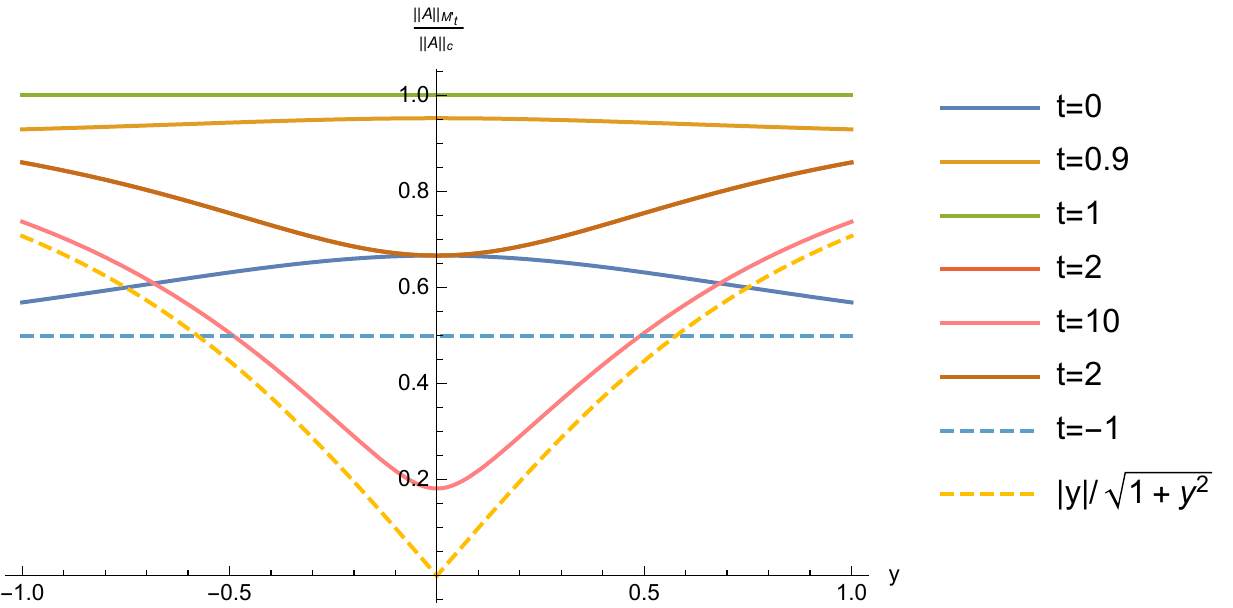}\qquad\qquad\qquad\qquad 
		\caption{The ratio $\|A\|_{M'_t} \, / \, \|A\|_{c}$ for $y\in[-1,1]$ and different value of $t$. We notice that the ration is always smaller than 1, except for $t=1$, which is the CHSH game.}
		\label{fig:modified-CHSH}
	\end{figure}

	\bigskip
	In the same spirit as the example above we shall now analyze another deformation of the CHSH game that was considered in \cite{lawson2010biased}. In the following we shall recall the game, and analyse it as the game considered before with the tools that we introduced. 
	
	The deformation of the CHSH game that was considered in \cite{lawson2010biased}, was given by
	$$G(p,q) = \begin{bmatrix}
		p\,q & p\,(1-q)\\
		q\,(1-p)& -(1-q)\,(1-p)
	\end{bmatrix},$$
	where $p,q\in[0,1]^2$. Note that this matrix is invertible for all $(p,q) \in (0,1)^2$.
	
	In the following we give the classical bias $\beta(G(p,q))$ for different $p,q\in[0,1]^2$.
	\begin{itemize}
		\item For $p$ and $q$ satisfying $p,q\geq \frac{1}{2}$, the classical bias of the game $\beta(G(p,q))$ is given by $$\beta(G(p,q))=1-2\,(1-p)\cdot(1-q)$$
		\item For $p$ and $q$ satisfying  $p\leq\frac{1}{2},q\geq \frac{1}{2}$, the classical bias of the game $\beta(G(p,q))$ is given by $$\beta(G(p,q))=1-2\,p\cdot(1-q)$$
		\item For $p$ and $q$ satisfying  $q\leq\frac{1}{2},p\geq \frac{1}{2}$, the classical bias of the game $\beta(G(p,q))$ is given by $$\beta(G(p,q))=1-2\,q\cdot(1-p)$$
		\item For $p$ and $q$ satisfying  $p,q\leq \frac{1}{2}$, the classical bias of the game $\beta(G(p,q))$ is given by
		$$\beta(G(p,q))=1-2p\cdot q$$
	\end{itemize}.
	\begin{remark}
		The classical bias of the game $\beta(G(p,q))$ for $p,q\geq \frac{1}{2}$ was already shown in \cite{lawson2010biased}.
	\end{remark}
	To give the classical bias of the game $\beta(G(p,q))$  for different $p,q\in[0,1]^2$, one shall use $\min(\cdot,\cdot)$ defined as 
	$$\min(x,y):=\frac{1}{2}(x+y-|x-y|).$$
	with $x,y\in\mathbb{R}$.
	
	One can easily check that for $p\in[0,1]$ we have
	$$\min(p,1-p)=\frac{1}{2}(1-|2p-1|)=\begin{cases}
		1-p\qquad \text{for}\quad p\geq\frac{1}{2}&\\
		p\qquad \text{for}\quad p\leq \frac{1}{2}&
	\end{cases}$$
	The same results hold for $\min(q,1-q)$ with $q\in[0,1]$.
	
	It can be easily seen that the classical bias of the game for $p,q\in[0,1]^2$ is given by:
	$$\beta(G(p,q))=1-2\cdot\min(p,1-p)\min(q,1-q).$$
	
	In our setting, we shall consider the normalised game $G'(p,q)$ for all $p,q\in[0,1]^2$
	$$G'(p,q)=\frac{1}{\beta(G(p,q))} \begin{bmatrix}
		p\,q & p\,(1-q)\\
		q\,(1-p)& -(1-q)\,(1-p)
	\end{bmatrix}. $$
	As in the example above we consider the following pair of spin observables
	$$A:=(\sigma_X,y\sigma_Y),$$
	where $y \in [-1,1]$ is a parameter we shall vary. In the following we will compute the $\|A\|_{G'(p,q)}$.
	\begin{align*}
		\|A\|_{G'_{p,q}}&:=\lambda_{\max}\bigg[\sum_{y=1}^2\bigg|\sum_{x=1}^2 G'(p,q)_{x,y}A_x\bigg|\bigg]\\
		&=\frac{1}{|\beta(G(p,q))|}\lambda_{\max}\Big(\Big|p\,q\,\sigma_X+y\,q\,(1-p)\sigma_Y\Big|+\Big|p\,(1-q)\sigma_X-y\,(1-p)\,(1-q)\,\sigma_Y\Big|\Big).\end{align*}
	A simple calculation shows that
	$$\|A\|_{G'(p,q)}=\frac{1}{|1-2\min(p,1-p)\min(q,1-q)|}\Big(\Big[p^2\,q^2+y^2\,q^2\,(1-p)^2\Big]^{\frac{1}{2}}+\Big[p^2\,(1-q)^2+y^2\,(1-p)^2\,(1-q)^2\Big]^{\frac{1}{2}}\Big).$$
	
	Note that for $p=1/2$, we have a simplification: 
	$$\|A\|_{G'(1/2,q)}=\frac{\sqrt{1+y^2}}{2\max(q,1-q)} = \frac{\|A\|_c}{2\max(q,1-q)}.$$
	
	We plot in Figure \ref{fig: Alice violation for G'(p,q)} the set of pairs $(p,q)$ such that $\|A\|_{G'(p,q)}>1$, that is the game parameter region where Alice observes a Bell violation for different values of $y\in[0,1]$. In Figure $\ref{fig: norme A_G'(p,q) }$ we plot the norm $\|A\|_{G'(p,q)}$ while in Figure $\ref{fig: Rapport norme A_G(p,q) et A_c}$ the ratio of $\|A\|_{G'(p,q)}/\|A\|_c$ for fixed value of $p$ and $q$ .

	\begin{figure}[htb!]
		\centering
		\includegraphics[width=0.8\textwidth]{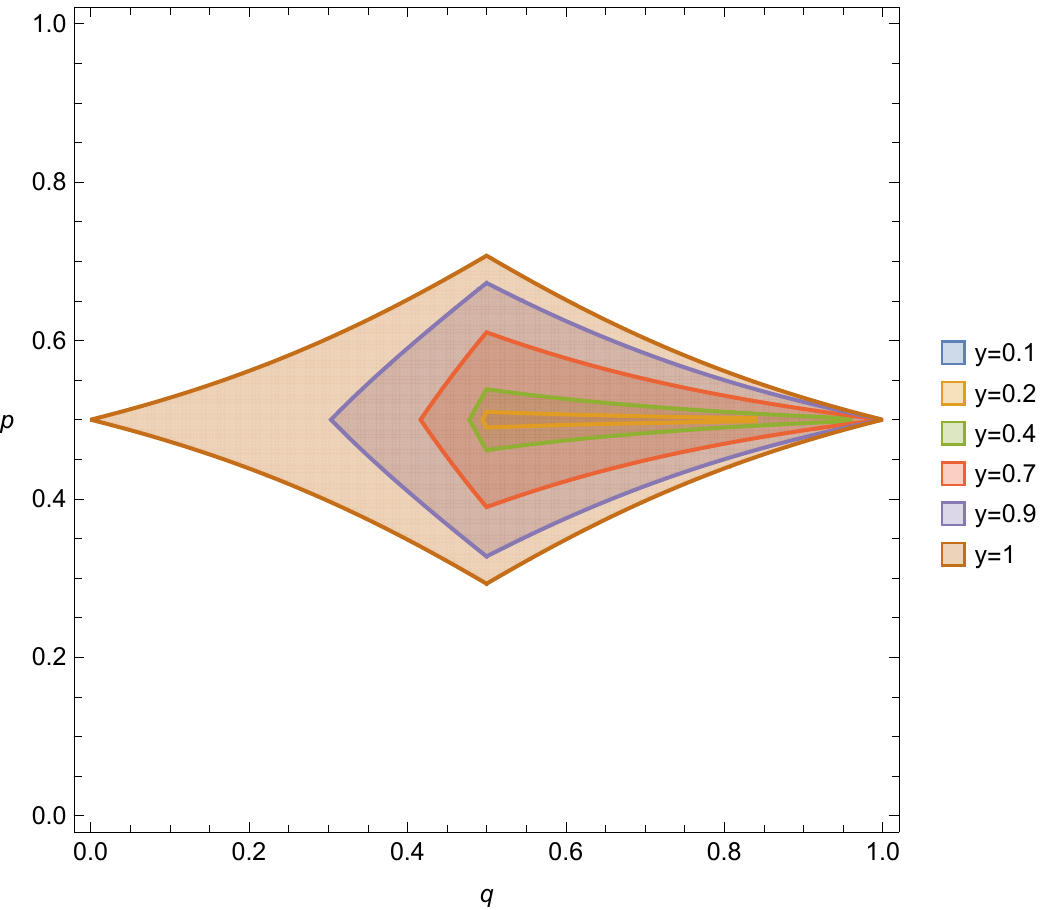}\qquad\qquad\qquad\qquad 
		\caption{The $p$,$q$ region where Alice observe a violation $\|A\|_{G'(p,q)}>1$ for different value of $y$.}
		\label{fig: Alice violation for G'(p,q)}
	\end{figure}
	
	\begin{figure}[htb!]
		\centering
		\includegraphics[width=\textwidth]{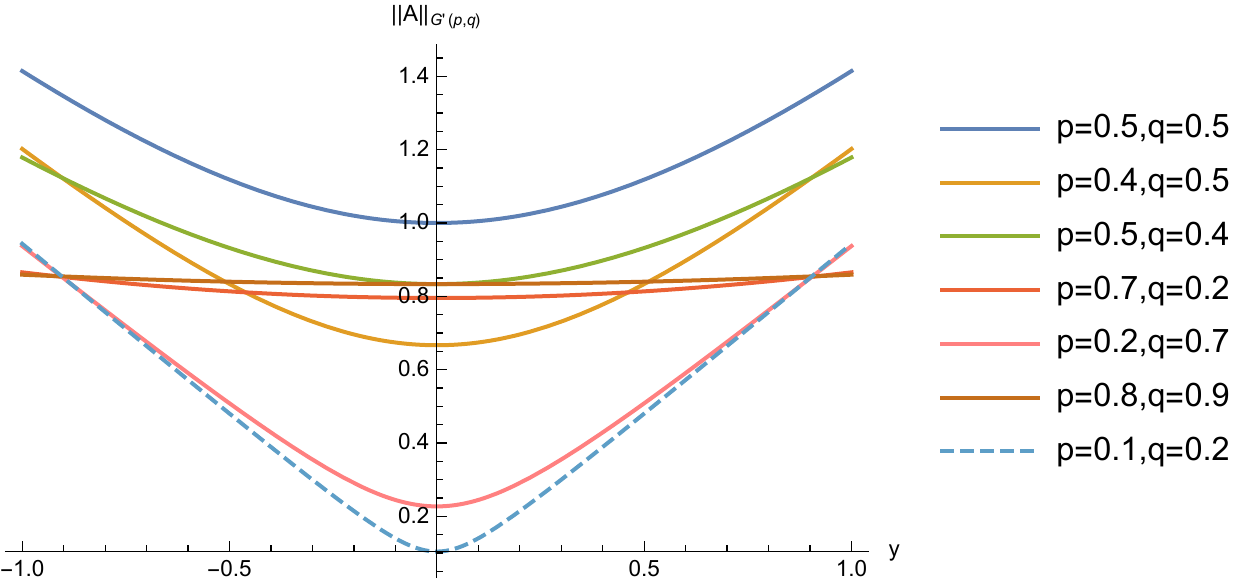}\qquad\qquad\qquad\qquad 
		\caption{The norm $\|A\|_{G'(p,q)}$ for $y\in[-1,1]$ and different value of $p$ and $q$.}
		\label{fig: norme A_G'(p,q) }
	\end{figure}

	\begin{figure}[htb!]
		\centering
		\includegraphics[width=\textwidth]{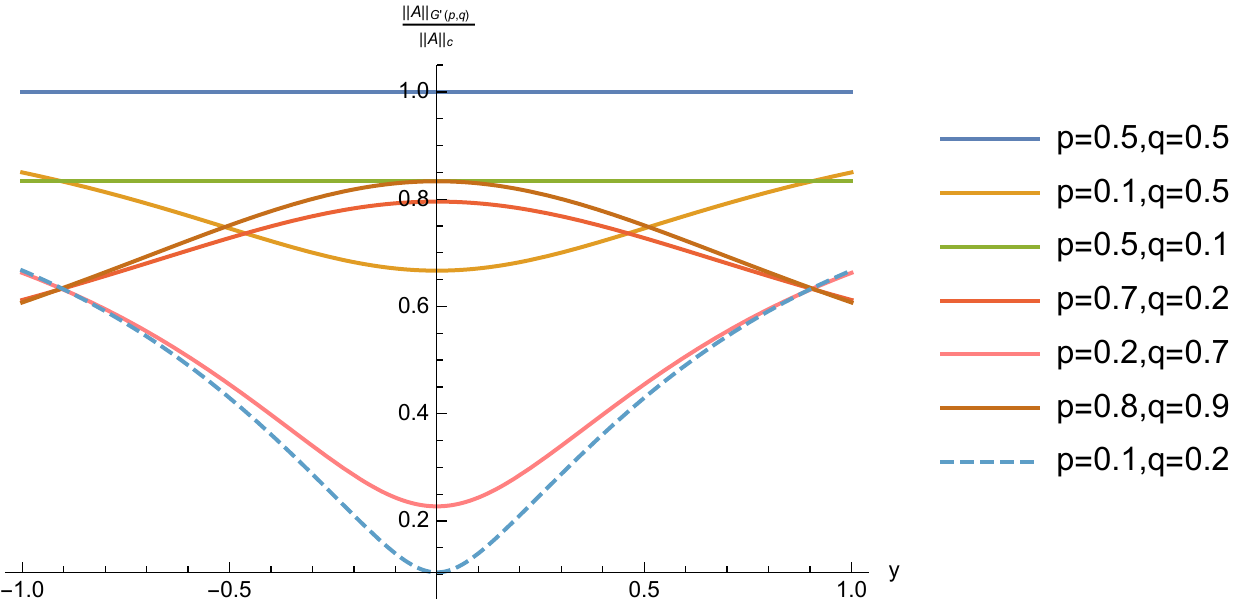}\qquad\qquad\qquad\qquad 
		\caption{The ratio $\|A\|_{G'(p,q)}/\|A\|_c$ for $y\in[-1,1]$ and different value of $p$ and $q$.}
		\label{fig: Rapport norme A_G(p,q) et A_c}
	\end{figure}

	\bigskip
	
	We now move on to the last example, the pure correlation part of the $I_{3322}$ tight Bell inequality (here, $N=3$):
	$$M_{\textrm{3322}} = \frac 1 4 \begin{bmatrix}
		1 & 1 & 1 \\
		1 & 1 & -1 \\
		1 & -1 & 0 
	\end{bmatrix}.$$
	The inverse of the matrix above has entries with absolute value 2, so our main result does not apply. Indeed, one can see that
	$$\|s(\sigma_X,\sigma_Y,\sigma_Z)\|_c \leq 1 \iff s \leq \frac{1}{\sqrt 3},$$
	while
	$$\|s(\sigma_X,\sigma_Y,\sigma_Z)\|_{M_{3322}} \leq 1 \iff s \leq \frac{4}{\sqrt 2 + 2\sqrt 3} > \frac{1}{\sqrt 3}.$$
	This shows that, for tensors in the positive direction $(\sigma_X,\sigma_Y,\sigma_Z)$, for parameter values 
	$$s \in \left( \frac{1}{\sqrt 3}, \frac{4}{\sqrt 2 + 2 \sqrt 3} \right],$$
	we have 
	$$\|s(\sigma_X,\sigma_Y,\sigma_Z)\|_{M_{3322}} \leq 1 < \|s(\sigma_X,\sigma_Y,\sigma_Z)\|_c,$$
	so there exist incompatible dichotomic Pauli measurements which do not violate the pure correlation $I_{3322}$ Bell inequality \cite{quintino2014joint}.
	
	\section{Non-local games which characterize incompatibility}\label{sec:norm-equality}
	
	Up to this point, we have seen the following two inequalities relating the $M$-Bell-locality norm $\|\cdot \|_M$ and the compatibility norm $\|\cdot\|_c$ of a tuple of dichotomic quantum measurements: 
	$$\|A\|_M\leq\|A\|_{c}\|M\|_{\ell_1^N \otimes_{\epsilon} \ell_1^N} \qquad \text{ and } \qquad \|A\|_c \leq \|A\|_M \|M^{-1}\|_{\ell^N_\infty \otimes_\epsilon \ell^N_\infty}.$$
	In this section, we ask for which (invertible) non-local games $M$, these two inequalities, used together, allow us to conclude that $\|\cdot \|_M = \|\cdot\|_c$. Such an equality would prove a strong equivalence of Bell inequality violations and incompatibility for the game $M$, in the spirit of \cite{wolf2009measurements}.
	
	First, note that,  for an invertible game $M$ and a non-zero tuple of measurements $A$, we have 
	$$\|A\|_M\leq\|A\|_{c}\|M\|_{\ell_1^N \otimes_\epsilon \ell_1^N}\leq \|A\|_M\|M^{-1}\|_{\ell_\infty^N \otimes_\epsilon \ell_\infty^N}\|M\|_{\ell_1^N \otimes_\epsilon \ell_1^N},$$
	hence 
	\begin{equation}\label{eq:uncertainty}
		\|M^{-1}\|_{\ell_\infty^N \otimes_\epsilon \ell_\infty^N}\|M\|_{\ell_1^N \otimes_\epsilon \ell_1^N}\geq 1.
	\end{equation}
	In order to deduce that $\|\cdot \|_M = \|\cdot\|_c$, one requires
	$$\beta(M) = \|M\|_{\ell_1^N \otimes_\epsilon \ell_1^N} = 1 \qquad \text{ and } \qquad \|M^{-1}\|_{\ell_\infty^N \otimes_\epsilon \ell_\infty^N} = 1.$$
	Up to rescaling, this is equivalent to requiring that the inequality \eqref{eq:uncertainty} should be saturated. We now study the equality case in \eqref{eq:uncertainty}, which can be seen as an ``uncertainty relation'' for the non-local game $M$.
	
	Let us first show that \eqref{eq:uncertainty} cannot be saturated for $N \geq 3$. We recall and use a definition from \cite{aubrun2020universal} to understand the ratio of $\|\cdot\|_{X\otimes_{\pi}Y}$and $\|\cdot\|_{X\otimes_{\epsilon}Y}$ for two given  Banach spaces $X$ and $Y$.
	\begin{definition}\cite{aubrun2020universal}\label{def 8.7}
		Given two finite-dimensional Banach spaces $X$ and $Y$. There will always exist a constant $1\leq C<\infty$ such that:
		$$\|\cdot\|_{X\otimes_{\epsilon}Y}\leq\|\cdot\|_{X\otimes_{\pi}Y}\leq C \|\cdot\|_{X\otimes_{\epsilon}Y} $$
		One denotes $\rho(X,Y)$ the smallest $C$ satisfying this inequality. Equivalently one has 
		$$\rho(X,Y)=\sup_{0\neq z\in X\otimes Y}\frac{\|z\|_{X\otimes_{\pi}Y}}{\|z\|_{X\otimes_{\epsilon}Y}}$$
	\end{definition}
	We recall one of the important properties of $\rho(X,Y)$ in the case of $\ell_1$ and $\ell_\infty$ spaces.
	\begin{proposition}\cite[Proposition 13]{aubrun2020universal}
		For all $N \geq 2$, we have
		$$\rho(\ell_1^N,\ell_1^N) = \rho(\ell_\infty^N,\ell_\infty^N) \leq \sqrt{2N}.$$
	\end{proposition}
	
	With the help of the definition of $\rho(X,Y)$ and proposition above, we can\footnote{We thank Carlos Palazuelos for the proof of the proposition.} improve the inequality \eqref{eq:uncertainty}.
	\begin{proposition}\label{prop:norm-inequality-M-inverse}
		Let $M$ a real and invertible matrix. Then one has
		$$\|M^{-1}\|_{\ell^N_{\infty}\otimes_{\epsilon}\ell^N_{\infty}}\|M\|_{\ell^N_1\otimes_{\epsilon}\ell^N_1}\geq \frac{N}{\rho(\ell_{\infty}^N,\ell_{\infty}^N)}\geq \sqrt{\frac N 2}\geq 1.$$
		for $N\geq 2$. In particular, for $N \geq 3$, the last inequality above is strict.
	\end{proposition}
	\begin{proof}
		Let $$N=\Tr[M^{-1}M]=\langle \tilde{M},M\rangle_{H.S}\leq \|M\|_{\ell^N_1\otimes_{\epsilon}\ell^N_1}\|\tilde{M}\|_{\ell^N_{\infty}\otimes_{\pi}\ell^N_{\infty}}$$
		Where $\tilde{M}:=(M^{-1})^T$
		Thus we have by definition \ref{def 8.7} and we recall $\langle A,B\rangle_{H.S}:=\Tr[A^*B] $ $$N\leq\|M\|_{\ell^N_1\otimes_{\epsilon}\ell^N_1}\|\tilde{M}\|_{\ell^N_{\infty}\otimes_{\epsilon}\ell^N_{\infty}}\rho(\ell^N_{\infty},\ell^N_{\infty})$$
		Thus we have $$1\leq\frac{N}{\rho(\ell^N_{\infty},\ell^N_{\infty})}\leq\|M\|_{\ell^N_1\otimes_{\epsilon}\ell^N_1}\|M^{-1}\|_{\ell^N_{\infty}\otimes_{\epsilon}\ell^N_{\infty}}.$$
	\end{proof}
	
	Having shown that inequality \eqref{eq:uncertainty} cannot be saturated for $N \geq 3$, we now focus on the $N=2$ case. We need the following lemma\footnote{We thank Zbigniew Pucha{\l}a for this result.}. 
	
	\begin{lemma}\label{uncertaincty prop}
		For any matrix $X\in\mathcal{M}_N(\mathbb{C})$ and for any unitary operators $U,V \in \mathcal{U}_N$, we have
		$$\|UXV^*\|_{\ell_\infty^N \otimes_\epsilon \ell_\infty^N}\|X\|_{\ell_\infty^N \otimes_\epsilon \ell_\infty^N}\geq \frac 1 N|\det X|^{2/N}.$$
		Equality holds if and only if both $X$ and $UXV^*$ are scalar multiples of Hadamard matrices. 
	\end{lemma}
	\begin{proof}
		Let $x \in \mathbb C^{N^2}$ be the vectorization of $X$; we have 
		$$\|X\|_{\ell_\infty^N \otimes_\epsilon \ell_\infty^N} = \max_{i,j \in [N]} |X_{ij}| =  \|x\|_{\ell_\infty^{N^2}}.$$
		Moreover, the vectorization of $UXV^*$ is given by
		$$y:= (U \otimes \bar V) x.$$
		Using the unitarity of $U,V$ and the fact that for all vectors $z \in \mathbb C^{N^2}$, $\|z\|_{\ell_\infty^{N^2}} \geq \|z\|_2/N$, we have 
		\begin{align*}
			\|UXV^*\|_{\ell_\infty^N \otimes_\epsilon \ell_\infty^N}\|X\|_{\ell_\infty^N \otimes_\epsilon \ell_\infty^N} &= \|y\|_{\ell_\infty^{N^2}}\|x\|_{\ell_\infty^{N^2}} \geq \frac{1}{N^2}\|y\|_2\|x\|_2 \\
			&= \frac{1}{N^2}\|x\|_2^2 = \frac{1}{N^2}\|X\|_2^2=\frac{1}{N^2}\sum_{i=1}^N\sigma_i(X)^2.
		\end{align*}
		Above, $\sigma_i(X)$ denote the singular values of $X$. Using now the arithmetic mean-geometric mean (AM-GM) inequality, we have
		$$\frac{1}{N}\sum_{i=1}^N\sigma_i(X)^2\geq\left(\prod_{i=1}^N\sigma_i(X)^2\right)^{\frac{1}{N}}=|\det X|^{\frac{2}{N}}.$$
		Hence 
		$$\|UXV^*\|_{\ell_\infty^N \otimes_\epsilon \ell_\infty^N}\|X\|_{\ell_\infty^N \otimes_\epsilon \ell_\infty^N}\geq \frac 1 N|\det X|^{\frac{2}{N}},$$
		proving the inequality.
		
		In the derivation above, we have used three inequalities: the lower bound on the $\ell_\infty$ norm of the vectors $x,y$ by their $\ell_2$ norms, and the arithmetic and geometric inequality. If the former, equality holds iff the entries of, respectively, $x$ and $y$ are flat; this corresponds to the matrices $X$ and $UXV^*$ having, respectively, entries of identical absolute values. The latter corresponds to the singular values of $X$ being identical, which corresponds to $X$ being a scalar multiple of a unitary matrix. The announced equality condition follows from these considerations.
	\end{proof}
	Recall that the Fourier matrix, also known as the discrete Fourier transform (DFT), is given by $$F = N^{-1/2} \Big[\omega^{ij}\Big]_{i,j=0}^{N-1},$$
	where $\omega = \exp(2 \pi \mathrm{i}/N)$. 
	
	\begin{proposition}\label{prop:equality-N-2}
		For $N=2$, all the invertible Bell functionals $M\in\mathcal{M}_2(\mathbb{R})$ satisfying $$\|M^{-1}\|_{\ell^2_{\infty}\otimes_{\epsilon}\ell^2_{\infty}}\|M\|_{\ell^2_1\otimes_{\epsilon}\ell^2_1} = 1$$
		are of the form
		$$M =a\begin{bmatrix}
			1 & 1 \\ 
			1 & -1 \end{bmatrix}$$
		with $a\in\mathbb R$, $a \neq 0$.
	\end{proposition}
	\begin{proof}
		Since $M\in\mathcal{M}_2(\mathbb R) $ one note that
		$$\|M^{-1}\|_{\ell^2_{\infty}\otimes_{\epsilon}\ell^2_{\infty}}=\frac{1}{|\det(M)|}\|M\|_{\ell^2_{\infty}\otimes_{\epsilon}\ell^2_{\infty}}$$
		and $$\|M\|_{\ell^2_1\otimes_{\epsilon}\ell^2_1}=\|(T\otimes T) M\|_{\ell^2_{\infty}\otimes_{\epsilon}\ell^2_{\infty}}=2\,\|(F\otimes F) M\|_{\ell^2_{\infty}\otimes_{\epsilon}\ell^2_{\infty}}$$
		where $T =\begin{bmatrix}
			1 & 1 \\ 
			1 & -1 \end{bmatrix} = \sqrt 2F$, is an isometry $T:\ell^2_1\to\ell^2_{\infty}$; geometrically, this corresponds to the fact that the unit ball of $\ell^2_1$ (a diamond) is a scaled rotation of the unit ball of $\ell^2_{\infty}$ (a square).\\
		Now using the lemma above one has 
		$$\|M^{-1}\|_{\ell^2_{\infty}\otimes_{\epsilon}\ell^2_{\infty}}\|M\|_{\ell^2_1\otimes_{\epsilon}\ell^2_1}=\frac{2}{|\det(M)|}\|M\|_{\ell^2_{\infty}\otimes_{\epsilon}\ell^2_{\infty}}\|(F\otimes F) M\|_{\ell^2_{\infty}\otimes_{\epsilon}\ell^2_{\infty}}\geq \frac{2}{|\det (M)|}\frac{|\det(M)|}{2}=1$$
		The equality holds as the lemma above iff $M$ is a unitary and the entries of $(F\otimes F)M$ and $M$ are flat (i.e.~have the same absolute value). Then $M$ is a multiple of a \emph{Hadamard matrix}:
		$$M =a\begin{bmatrix}
			1 & 1 \\ 
			1 & -1 \end{bmatrix} \qquad  a\in \mathbb{R},\, a\neq 0.$$
	\end{proof}

	Gathering Propositions \ref{prop:norm-inequality-M-inverse} and \ref{prop:equality-N-2}, we obtain the following important characterization of non-local games $M$ achieving equality in \eqref{eq:uncertainty}. 
	
	\begin{theorem}
		The only invertible non-local games $M \in \mathcal M_N(\mathbb R)$ satisfying 
		$$\|M^{-1}\|_{\ell^2_{\infty}\otimes_{\epsilon}\ell^2_{\infty}}\|M\|_{\ell^2_1\otimes_{\epsilon}\ell^2_1} = 1$$
		have two questions ($N=2$) and are variants of the CHSH game: $M = a M_{\textrm{CHSH}}$ for some $a \neq 0$.
	\end{theorem}
	
	Note that saturating inequality \eqref{eq:uncertainty} is just a \emph{sufficient} condition for having $\|\cdot \|_M = \|\cdot \|_c$. We leave the general case open: for which non-local games $M$, does one have $\|\cdot \|_M = \|\cdot \|_c$? 
	
	\section{Conclusion}
	Two of the most fundamental features of quantum mechanics are measurement incompatibility and the non-locality of correlations. In this work, we address the relation between the two concepts within the natural framework of  tensor norms. It was well-known that in order to observe correlation non-locality in a Bell-type experiment, one should use incompatible measurements. Moreover, it was shown that in some particular cases, such as the CHSH game, incompatibility and the violation of the Bell inequality are equivalent. In the current paper, we introduced a natural framework in which one can directly compare the two notions. We have shown that the incompatibility is not in general equivalent to the non-locality, by comparing two tensor norms. 
	
	Finally, let us address some questions we have left open. First and foremost, our setting is only adapted to dichotomic (2-outcome) POVMs; it would be interesting to extend the results in this paper to measurements with an arbitrary number of outcomes. The main obstacle here is encoding the outcomes of the $g$ POVMs with more than 2 outcomes into a relevant tensor. In Section \ref{sec:norm-equality}, we have shown that the two tensor norms, the one associated to a non-local game and the one associated to compatibility, cannot be shown to be equal using a simple chain of inequalities (except in the case of the CHSH game and its variants). The question whether the two norms can be equal (using other methods) is open. Here, one would need to rely on a more general argument instead of relying on some particular inequalities. Finally, our methods only cover XOR games with pure correlation terms; associating a tensor norm to more general games (such as the full $I_{3322}$ game) is an interesting open problem. 
	\bigskip
	
	\noindent\textit{Acknowledgements.} The authors would like to thank Carlos Palazuelos and Zbigniew Pucha{\l}a for help with the proofs of some results in Section \ref{sec:norm-equality}. The two anonymous referees provided us with an extensive list of suggestions and corrections; we are very thankful for their interest in our work, helping us greatly improve the presentation of our results. 
	
	The authors were supported by the ANR project \href{https://esquisses.math.cnrs.fr/}{ESQuisses}, grant number ANR-20-CE47-0014-01, and by the PHC programs \emph{Sakura} (Random Matrices and Tensors for Quantum Information and Machine Learning), \emph{Star} (Applications of random matrix theory and abstract harmonic analysis to quantum information theory), and \emph{Procope} (Entanglement Preservation in Quantum Information Theory). I.N.~was also supported by the ANR project \href{https://www.math.univ-toulouse.fr/~gcebron/STARS.php}{STARS}, grant number ANR-20-CE40-0008.

	\bibliographystyle{alpha}
	\bibliography{refs}
	
\end{document}